\documentclass[table,11pt,a4paper]{scrartcl}

\usepackage{amssymb,amsmath,amsthm,fullpage,enumitem,url}
\usepackage{complexity}
\usepackage{subcaption}

\providecommand{\texorpdfstring}[2]{#1}
\usepackage{csquotes}
\renewcommand{\mkbegdispquote}[2]{\itshape}

\usepackage{tikz}
\usetikzlibrary{arrows}
\usepackage{hyperref}
\usepackage[capitalise]{cleveref}
\usepackage[colorinlistoftodos]{todonotes}
\newcommand{\dmax}{d_{\mathsf{max}}}
\newcommand{\davg}{d_{\mathsf{avg}}}
\newcommand{\dmin}{d_{\mathsf{min}}}
\newcommand{\ct}{\mathcal{T}}

\renewcommand{\Pr}[1]{\mathbb{P}\left[\,#1\,\right]}
\newcommand{\Pru}[2]{\mathbb{P}_{#1}\!\left[#2\right]}
\newcommand{\Ex}[1]{\mathbb{E} \left[\,#1\,\right]}
\newcommand{\Exu}[2]{\mathbb{E}_{#1} \left[\,#2\,\right]}

\newcommand{\BO}[1]{\mathcal{O}\!\left(#1\right)} 
\newcommand{\BT}[1]{\Theta\!\left(#1\right)}
\newcommand{\lo}[1]{o\!\left(#1\right)}

\newcommand*{\abs}[1]{\lvert #1\rvert}

\newcommand*{\bfrac}[2]{\genfrac{(}{)}{}{}{#1}{#2}}

\newcommand{\Heb}[2]{H_{#1}^{\eps\mathsf{B}}(#2)}
\newcommand{\ETBcov}[2]{C_{#1}^{\eps\mathsf{TB}}(#2)}
\newcommand{\tetb}{t_{\mathsf{cov}}^{\eps\mathsf{TB}}}
\newcommand{\tebh}{t_{\mathsf{hit}}^{\eps\mathsf{B}}}
\newcommand{\tetp}{t_{\mathsf{cov}}^{\eps\mathsf{TB}}}
\newcommand{\tmix}{t_{\mathsf{mix}}}
\newcommand{\thit}{t_{\mathsf{hit}}}
\newcommand{\trel}{t_{\mathsf{rel}}}
\newcommand{\tsep}{t_{\mathsf{sep}}}

\renewcommand{\leq}{\leqslant}
\renewcommand{\geq}{\geqslant}

\renewcommand{\tilde}{\widetilde}

\newcommand{\eps}{\varepsilon}
\newcommand{\N}{\mathbb N}

\allowdisplaybreaks

\newtheorem{theorem}{Theorem}[section]
\newtheorem{lemma}[theorem]{Lemma}
\newtheorem{corollary}[theorem]{Corollary}
\newtheorem{proposition}[theorem]{Proposition}
\newtheorem{conjecture}[theorem]{Conjecture}

\newtheorem*{claim}{Claim}
\newtheorem{clm}{Claim}
\newtheorem{remark}[theorem]{Remark}

\newtheorem{quest}{Question}[section]
\newenvironment{poc}{\begin{proof}[Proof of claim]}{\end{proof}}

\theoremstyle{definition}
\newcommand{\gadgetname}{}
\newtheorem*{generic}{\gadgetname}
\newenvironment{gadget}[1]{\renewcommand{\gadgetname}{#1}\begin{generic}}{\end{generic}}

\usepackage{authblk}

\newcommand{\Quincrux}[4]{
	\begin{scope}[shift={(#2,#3)}]
		\filldraw[blue!30]
		({90}      : #1  ) -- ({210}  :#1) --
		({330} :          #1)  -- cycle;
		\draw[thick,blue] ({90}      : #1  ) -- ({210}  :#1);
		\draw[thick,blue] ({210}      : #1  ) -- ({330}  :#1);
		\draw[thick,blue] ({330}      : #1  ) -- ({90}  :#1);
		\draw[fill] ({90} :  #1) circle (.1*#1);
		\draw[fill] ({210} : #1) circle (.1*#1);
		\draw[fill] ({330} : #1 ) circle (.1*#1);
		\draw (0,.25) node[anchor=north]{#4};
	\end{scope}
}
\newcommand{\CascadeQuincruxR}[4]{
	\begin{scope}[shift={(#2,#3)}]
		\filldraw[blue!30]
		({120}      : #1  ) -- ({240}  :#1) --
		({0} :          #1)  -- cycle;
		\draw[thick,blue] ({120}      : #1  ) -- ({240}  :#1);
		\draw[thick,blue] ({240}      : #1  ) -- ({0}  :#1);
		\draw[thick,blue] ({0}      : #1  ) -- ({120}  :#1);
		\draw[fill] ({120} :  #1) circle (.1*#1);
		\draw[fill] ({240} : #1) circle (.1*#1);
		\draw[fill] ({0} : #1 ) circle (.1*#1);
		\draw (0,.25) node[anchor=north]{#4};
	\end{scope}
}

\newcommand{\CascadeQuincruxL}[4]{
	\begin{scope}[shift={(#2,#3)}]
		\filldraw[blue!30]
		({180}      : #1  ) -- ({300}  :#1) --
		({60} :          #1)  -- cycle;
		\draw[thick,blue] ({180}      : #1  ) -- ({300}  :#1);
		\draw[thick,blue] ({300}      : #1  ) -- ({60}  :#1);
		\draw[thick,blue] ({60}      : #1  ) -- ({180}  :#1);
		\draw[fill] ({180} :  #1) circle (.1*#1);
		\draw[fill] ({300} : #1) circle (.1*#1);
		\draw[fill] ({60} : #1 ) circle (.1*#1);
		\draw (0,.25) node[anchor=north]{#4};
	\end{scope}
}

\newcommand{\RoundClause}[4]{
	\begin{scope}[shift={(#2,#3)}]
		\def \triscale {.4}
		\filldraw[green!40]
		({30}      :                 #1  ) arc
		({30}      : {60} :#1) --
		({30} :                 #1+\triscale*#1)  -- cycle;
		\draw[thick,green]  ({60}      :                 #1  )-- ({30} :                #1+\triscale*#1);
		\draw[thick,green]  ({30}      :                 #1  )-- ({30} :                #1+\triscale*#1);
		\filldraw[green!50]
		({150}      :                  #1   ) arc
		({150}      : {180} :  #1 ) --
		({150} :                  #1+\triscale*#1)  -- cycle;
		\draw[thick,green]  ({180}      :                  #1   )-- ({150} :                 #1+\triscale*#1);
		\draw[thick,green]  ({150}      :                 #1   )-- ({150} :                  #1+\triscale*#1);
		
		\filldraw[green!50]
		({270}      :                  #1   ) arc
		({270}      : {300} : #1 ) --
		({270} :                  #1+\triscale*#1)  -- cycle;
		\draw[thick,green]  ({300}      :                 #1   )-- ({270} :                 #1+\triscale*#1);
		\draw[thick,green]  ({270}      :                #1   )-- ({270} :                 #1+\triscale*#1);
		\draw [thick,green,fill=green!20] ( #1 ,0) arc[start angle=0, end angle=360,radius= #1 ];
		\draw[fill] ({30} : #1) circle (.04*#1);
		\draw[fill] ({30} : #1+\triscale*#1) circle (.04*#1);
		\draw[fill] ({150} : #1) circle (.04*#1);
		\draw[fill] ({150} : #1+\triscale*#1) circle (.04*#1);
		\draw[fill] ({270} : #1) circle (.04*#1);
		\draw[fill] ({270} : #1+\triscale*#1) circle (.04*#1);
		\draw (0,.25) node[anchor=north]{#4};
	\end{scope}
}

\newcommand{\ReflectedRoundClause}[4]{
	\begin{scope}[shift={(#2,#3)}]
		\def \triscale {.4}
		\filldraw[green!40]
		({30}      :                 #1  ) arc
		({30}      : {0} :#1) --
		({30} :                 #1+\triscale*#1)  -- cycle;
		\draw[thick,green]  ({0}      :                 #1  )-- ({30} :                #1+\triscale*#1);
		\draw[thick,green]  ({30}      :                 #1  )-- ({30} :                #1+\triscale*#1);
		\filldraw[green!50]
		({150}      :                  #1   ) arc
		({150}      : {120} :  #1 ) --
		({150} :                  #1+\triscale*#1)  -- cycle;
		\draw[thick,green]  ({120}      :                  #1   )-- ({150} :                 #1+\triscale*#1);
		\draw[thick,green]  ({150}      :                 #1   )-- ({150} :                  #1+\triscale*#1);
		
		\filldraw[green!50]
		({270}      :                  #1   ) arc
		({270}      : {240} : #1 ) --
		({270} :                  #1+\triscale*#1)  -- cycle;
		\draw[thick,green]  ({240}      :                 #1   )-- ({270} :                 #1+\triscale*#1);
		\draw[thick,green]  ({270}      :                #1   )-- ({270} :                 #1+\triscale*#1);
		\draw [thick,green,fill=green!20] ( #1 ,0) arc[start angle=0, end angle=360,radius= #1 ];
		\draw[fill] ({30} : #1) circle (.04*#1);
		\draw[fill] ({30} : #1+\triscale*#1) circle (.04*#1);
		\draw[fill] ({150} : #1) circle (.04*#1);
		\draw[fill] ({150} : #1+\triscale*#1) circle (.04*#1);
		\draw[fill] ({270} : #1) circle (.04*#1);
		\draw[fill] ({270} : #1+\triscale*#1) circle (.04*#1);
		\draw (0,.25) node[anchor=north]{#4};
	\end{scope}
}

\date{}
\begin{document}
	\providecommand{\keywords}[1]{\noindent\textbf{Keywords:} #1.}
	\providecommand{\MSC}[1]{\noindent\textbf{AMS MSC 2010:} #1.}
	\title{Time Dependent Biased Random Walks\footnotetext{Some results from this paper appeared in \newblock {\em The 11th  Innovations  in  Theoretical  Computer  Science Conference (ITCS 2020)}, volume 151 of LIPIcs, pages 76:1--76:19 \cite{ITCSpaper}}}
	\author[1]{John Haslegrave}
	\author[2]{Thomas Sauerwald}
	\author[2]{John Sylvester}
	\affil[1]{Mathematics Institute, University of Warwick}
	\affil[2]{Department of Computer Science \& Technology, University of Cambridge}
	\maketitle
	\begin{abstract}
		
		We study the biased random walk where at each step of a random walk a ``controller'' can, with a certain small probability, move the walk to an arbitrary neighbour. This model was introduced by Azar et al.\ [STOC'1992]; we extend their work to the time dependent setting and consider cover times of this walk. We obtain new bounds on the cover and hitting times. Azar et al.\ conjectured that the controller can increase the stationary probability of a vertex from $p$ to $p^{1-\eps}$; while this conjecture is not true in full generality, we propose a best-possible amended version of this conjecture and confirm it for a broad class of graphs.
		We also consider the problem of computing an optimal strategy for the controller to minimise the cover time and show that for directed graphs determining the cover time is $\PSPACE$-complete.
	\end{abstract}
	
			
	\MSC{05C81, 60J10, 68R10, 68Q17}
	
	\section{Introduction}
	
	Randomised algorithms have come to occupy a central place within theoretical computer science and had a profound affect on the development of algorithms and complexity theory \cite{Karp91,MotRag}. Most randomised algorithms assume access to a source of unbiased independent random bits. In practice, however, truly independent unbiased random bits are inconvenient, if not impossible, to obtain. We can generate pseudo-random bits on a computer fairly effectively \cite{GoldPseudo} but if computational resources are constrained the quality of these bits may suffer; in particular they may be biased or correlated. Another reason to consider the dependency of randomised algorithms on the random bits they use, other than imperfect generation, is that an adversary may seek to tamper with a source of randomness to influence the output of a randomised algorithm. This raises the natural question of whether relaxing the unbiased and independent assumptions have a notable effect on the efficacy of randomised algorithms. This is a question many researchers have studied since early in the development of randomised algorithms \cite{AlonBiasedCoin,BopNarCoin,VazRandPoly}.

	Motivated by this question Azar, Broder, Karlin, Linial and Phillips \cite{ABKLPbias} introduced the $\eps$-biased random walk ($\eps$-BRW). This process is a walk on a graph where at each step with probability $\eps$ a controller can choose a neighbour of the current vertex to move to, otherwise a uniformly random neighbour is selected. One can see this process from two different perspectives. The first interpretation is to see the model as a simple random walk (SRW) with some adversarial noise. That is, the SRW moves to uniformly and independently sampled neighbour in each time step, however there is an adversarial controller who with probability $\eps$ can change the random bits used to sample the next step to their advantage. In particular one may consider this as a basic model for any randomised algorithm which uses their random bits to find a correct solution. A more specific class of examples is the search for a witness to the truth of a given statement, for example \cite{UniTrans,SoSt}, where the objective of the adversary may be to prevent us finding a witness. Other concrete instances are Pollard's Rho and Kangaroo algorithms for solving the discrete logarithm problem \cite{BirthdayParadoxPollard,Kangaroo,Pollard}. The second perspective sees the controller as an advisor who guides an agent to some objective, however with probability $1-\eps $ at each step their signal is lost and agent moves to a neighbouring vertex sampled at random. In this setting the $\eps$-BRW and very closely related models have be studied in the contexts of graph searching with noisy advice \cite{Binsearh,LocatingTarget,SearchNoise} and collective navigation \cite{antblazed}. Another work \cite{Navigate} studies the problem of searching for an adversarially placed target in a tree where a ``signpost'' pointing toward the target appears at each vertex with probability $\eps$. Thus in \cite{Navigate} the controller hints appear randomly in space, as opposed to randomly at steps as in $\eps$-BRW.

	Azar et al.~\cite{ABKLPbias} consider a pair of different objectives for the controller. The first objective is to maximise/minimise weighted sums of stationary probabilities, where they obtain bounds on how much the controller can influence the stationary probabilities of certain graph classes. This fits most naturally with the first perspective on the walk. The second controller objective they study is that of minimising the expected hitting time of a given set of vertices. This fits more naturally with the second perspective however, one motivation for obtaining bounds on the hitting times is that they can also be used to bound stationary probabilities (via return times). They show that optimal strategies for maximising or minimising stationary probabilities or hitting times can be computed in polynomial time. Since finding an optimal controller strategy for either objective can be cast as a Markov decision process (MDP) \cite{Derman}, it follows that there exists an optimal strategy for these tasks which is independent of time. Consequently, Azar et al.\ only consider fixed strategies which are independent of time. 
	
	We extend the work of Azar et al.\ \cite{ABKLPbias} by studying the cover time of $\eps$-biased random walks, which is the expected time for the walk to visit every vertex of the graph. In the setting of memoryless search algorithms with noisy advice, the cover time of a time dependent $\eps$-biased walk is the run time of the probabilistic following algorithm \cite[Supplementary material]{antblazed} applied to the task of collecting a token from every vertex (as opposed to just one token as in \cite{Binsearh,LocatingTarget,SearchNoise}).

	\subsection{Our Results}

	In \cref{S:regular} we introduce a new method which is crucial to cope with the time dependencies of the $\epsilon$-TBRW. We first consider a ``trajectory-tree'' which encapsulates all walks of a given length from a fixed start vertex in a connected graph $G$ by embedding them into a rooted tree. We also introduce a symmetric operator on real vectors which describes the action of the $\eps$-TBRW. The combination of the operator and trajectory-tree allows to us to show that the $\eps$-TBRW can significantly increase the probabilities of rare events described by trajectories, that is:
	\begin{enumerate}[label=(\arabic*)]
		\item Let $u\in V$, $t > 0$, $0\leq  \eps \leq  1$ and $S$ be a subset of trajectories of length $t$ from $u$. Let $p$ be the probability the SRW samples a trajectory from $S$. Then a controller can increase the probability of being in $S$ after $t$ steps from $u$ from $p$ to $p^{1-\eps}$. (See \cref{nonregboostnew}.)
	\end{enumerate}
	This result can be applied to bound cover and hitting times in terms of the number of vertices $n$, the minimum and average degrees $\dmin$ and $\davg $,  and finally $\thit$, $t_{\mathsf{rel}}$ and $\tmix$ which are the hitting, relaxation and mixing times of the lazy random walk; see Section \ref{formaldef} for full definitions.  
	\begin{enumerate}[resume*]
		\item For any vertex $u$ there is a strategy so that the $\eps$-TBRW started from $u$ covers $G$ in expected time at most
		\[\BO{\frac{\thit}{\eps}\cdot \log\left( \frac{\davg\cdot \trel \cdot \log  n}{\dmin} \right) }.\]
		It should be noted that, for regular graphs, this upper bound breaks the lower bound of $\Omega(n \log n)$ for the cover time of simple random walks
if $t_{\mathsf{rel}} = o( (\log n/ \log \log n)^2)$; in particular, for expanders we obtain a nearly-optimal cover time bound of $O(n \log \log n)$.
		\item For any two vertices $u,v\in V$ there is a strategy so that for the $\eps$-TBRW the expected time to reach $v$ from $u$ is at most \[\BO{ \left(\frac{n\cdot \davg}{ d_{\mathsf{min}}}\right)^{1-\eps}\cdot \left( t_{\mathsf{mix}}\right)^{\frac{2+\eps}{3}}}.\]
	\end{enumerate}
	(See Theorems \ref{trelbdd} and \ref{trelhit} for the two results above.)	
	
	In \cref{AzarConjSec} we study how much the controller can affect the stationary distribution of any vertex in our graph. Azar et al.\ \cite{ABKLPbias}
	introduced this problem and showed that for any bounded degree regular graph a controller can increase the stationary probability of any vertex from $p$ to $p^{1-\Omega(\eps)} $. By applying the results from \cref{S:regular} we prove a stronger bound for graphs with small relaxation time and sub-polynomial degree ratio:
	\begin{enumerate}[resume*]
		\item In any graph a controller can increase the stationary probability of any vertex from $p$ to $p^{1-\eps+\delta} $, where $\delta=\ln\left( 16\cdot \tmix  \right)/\abs{\ln p} $. (See \cref{azarconj}.)
	\end{enumerate}
Azar et al.\ \cite{ABKLPbias} conjectured that in any graph the controller can boost the stationary probability of any chosen vertex from $p$ to $p^{1-\eps}$ (see Conjecture~\ref{abklp}), thus we confirm their conjecture (up to a negligible error in the exponent) for the class of graphs above, including expanders.

	Motivated by this conjecture and a comment of Azar et al.\ stating that for regular graphs the interesting case is when $\eps$ is not substantially larger than $1/\dmax$ (the reciprocal of the maximum degree). We try to quantify the effect of a controller in this regime. Establishing several bounds and counter-examples we reveal the following trichotomy in terms of the density of the graph:
	 
	\begin{enumerate}[resume*]
		
	\item For any graph with $\dmax=  \lo{\log n/\log\log n}$, a controller for the $\eps$-BRW can increase the stationary probability of any vertex by more than a constant factor. 
		\item For any graph which is everywhere dense, i.e., has a minimum degree of $\Omega(n)$, a controller cannot increase any entry in the stationary distribution by more than a constant factor.  
	\item However, any polynomial but sublinear degree regime contains regular graphs for which entries in the stationary distribution can be increased by a polynomial factor, but for almost all almost-regular graphs, no entry can be increased by more than a constant factor.
	\end{enumerate}
(See \cref{cor-small-poly}, and Propositions \ref{prop:dense}, \ref{GnpNoBoost} and \ref{RegCycleBoost}  respectively for the above results.)
	
	In \cref{complexsec} we consider the complexity of finding an optimal strategy to cover a graph in minimum expected time. Azar et al.\ considered this problem for hitting times and showed that there is a polynomial algorithm to determine an optimal strategy on directed graphs; we establish a dichotomy by proving complexity theoretic lower bounds for the cover time. 
	\begin{enumerate}[resume*]
		\item The problems of deciding between two neighbours as the next step in order to minimise the cover time, and deciding if the cover time of a vertex subset is less than a given value are both $\PSPACE$-complete on directed graphs. (See Theorems \ref{covinPSPACE} and \ref{allhard}.)
	\end{enumerate}
Adapting previous results for the related choice random walk process \cite{ITCSpaper}, we also conclude:
	\begin{enumerate}[resume*]
		\item The two problems mentioned above in (8) are $\NP$-hard on undirected graphs. (See Theorem \ref{NextIsNPHard}.) 
\end{enumerate}
	Finally in \cref{Conclude} we  conclude with some open problems and conjectures. 
	
	\section{Preliminaries} \label{formaldef}
	We shall now formally describe the $\eps$-biased and $\eps$-time-biased random walk model and introduce some notation.
	Throughout this paper we shall always consider a connected $n$-vertex simple graph $G=(V,E)$, which unless otherwise specified, will be  unweighted. We write $\Gamma(v)$ for the neighbourhood of a vertex $v$ and call $d(v)=|\Gamma(v)|$ the degree of $v$. We use $\dmax$, $\dmin$ and $\davg$ to denote the maximum, minimum and average degrees of a graph respectively. Given a Markov chain $\mathbf{H}=(h_{x,y})_{x,y\in V}$ with transition probabilities $h_{x,y}$, let $h_{x,y}^{(t)}$ denote the probability the walk started at state $x$ is at $y$ after $t$ steps. Let $\pi_{\mathbf{H}}$ denote the stationary distribution of $\mathbf{H}$, and throughout we let $\pi=\pi_{\mathbf{P}}$  where $\mathbf{P}$ is the transition matrix of a simple random walk (SRW), thus $\mathbf{P}=(p_{x,y})_{x,y\in V}$ where $p_{x,y}=1/d(x)$ if $xy\in E$ and $0$ otherwise.
		
	Azar et al.\ \cite{ABKLPbias}, building on earlier work \cite{ben1987collective}, introduced the $\eps$-biased random walk ($\eps$-BRW) on a graph $G$. Each step of the $\eps$-BRW is preceded by an $(\eps, 1 - \eps)$-coin flip. With probability $1 -\eps$ a step of the simple random walk is performed, but with probability $\eps$ the controller gets to select which neighbour to move to. The selection can be probabilistic, but it is time independent. Thus if $\mathbf P$ is the transition matrix of the simple random walk, then the transition matrix $\mathbf Q^{\eps\text{B}}$ of the $\eps$-biased random walk is given by
	\begin{equation}\label{bias}\mathbf Q^{\eps\text{B}} = (1 - \eps)\mathbf P + \eps\mathbf B,\end{equation}
	where $\mathbf B$ is an arbitrary stochastic matrix chosen by the controller, with
	support restricted to $E(G)$. The controller of an $\eps$-BRW has full knowledge of $G$.
	
	Azar et al.\ focused on the problems of bias strategies which either minimise or maximise the stationary probabilities of sets of vertices or which minimise the hitting times of vertices. Azar et al.\ \cite[Sec.\ 4]{ABKLPbias} make the connection between Markov decision processes and the $\eps$-biased walk; in particular they observe that the two tasks they study can be identified as the expected average cost and optimal first-passage problems respectively in this context \cite{Derman}. As a result of this, the existence of time independent optimal strategies for both objectives follow from Theorems $2$ and $3$ respectively in \cite[Ch.\ 3]{Derman}. For this reason Azar et al.\ restrict to the class of unchanging strategies, where we say that an $\eps$-bias strategy is \textit{unchanging} if it is independent of both time and the history of the walk.
	
	It is clear that if we wish to consider optimal strategies to cover a graph (visit every vertex) in shortest expected time then we must include strategies which depend on the set of vertices already visited by the walk. Let $\mathcal{H}_t$ be the history of the random walk up to time $t$, that is the sigma algebra $\mathcal{H}_t= \sigma\left(X_0, \dots,X_t \right)$ generated by all steps of the walk up to and including time $t$. Thus we consider a time-dependent version, where the bias matrix $\mathbf B_t$ may depend on the time $t$ and the history $\mathcal{H}_t$; we refer to this as the $\eps$-time-biased walk ($\eps$-TBRW).
	
	Let $\ETBcov{v}{G}$ denote the minimum expected time (taken over all strategies) for the $\eps$-TBRW  to visit every vertex of $G$ starting from $v$, and define the \textit{cover time} $\tetb(G):=\max_{v\in V}\ETBcov vG$. Similarly let $\Heb xy$ denote the minimum expected time for the $\eps$-biased walk to reach $y$, which may be a single vertex or a set of vertices, starting from a vertex $x$. We do not need to provide notation for the hitting times of the $\eps$-TBRW since, as mentioned before, there is always a time-independent optimal strategy for hitting a given vertex \cite[Thm.\ 11]{ABKLPbias}, thus hitting times in the $\eps$-TBRW and $\eps$-BRW are the same. We also define the \textit{hitting time} $\tebh(G): = \max_{x,y\in V}\Heb xy$. Any unchanging strategy of the $\eps$-BRW on a finite connected graph results in an irreducible Markov chain $\mathbf{Q}$ and thus, when appropriate, we refer to its stationary distribution as $\pi_{\mathbf{Q}}$. 
	
	  Let $\mathbf{I}$ denote the identity matrix. Given a Markov chain $\mathbf{H}$ we call $\tilde{\mathbf{H}}=(\mathbf{I}+\mathbf{H})/2$ the \textit{lazy chain} of $\mathbf{H}$, and note that $\pi_{\tilde{\mathbf{H}}}=\pi_{\mathbf{H}}$. One important case is $\tilde{\mathbf{P}}$, where $\mathbf{P}$ is the SRW, we refer to this as the lazy random walk (LRW). Let $1=\lambda_1> \lambda_2\geq \cdots\geq \lambda_n\geq -1 $ be the eigenvalues of a simple random walk (SRW) on a connected $n$ vertex graph $G$ and define $\lambda_* =\max\left\{|\lambda_i| : i = 2, \dots, n \right\}$. Let $t_{\mathsf{rel}} := (1-\tilde{\lambda}_2)^{-1}$ be the relaxation time of $G$, where $\tilde{\lambda}_2$ is the second largest eigenvalue of $\tilde{\mathbf{P}}$, the LRW on $G$.  We let \[\tmix=\min_{t\geq 1}\left\{t:\;\max_{x\in V}||\tilde{p}_{x,\cdot }^{(t)}- \pi(\cdot)||_{\mathsf{TV}}\leq 1/4\right\}, \quad \text{where }\quad ||\tilde{p}_{x,\cdot }^{(t)}- \pi(\cdot)||_{\mathsf{TV}}= \frac{1}{2}\sum_{y\in V}\left|\tilde{p}^{(t)}_{x,y}-\pi(y) \right|, \] denote the \textit{total variation mixing time} of $G$.  For the lazy random walk (LRW) $\tilde{\mathbf{P}}$ we define   \begin{equation}\label{eq:sep}\tsep= \inf\left\{t : \max_{x,y\in V}\left[1-\frac{\tilde{p}^{(t)}_{x,y}}{\pi(y)}\right]\geq \frac{1}{\mathrm{e}}\right\}, \qquad \text{and}\qquad t_{\infty}= \inf\left\{t : \max_{x,y\in V}\left|\frac{\tilde{p}^{(t)}_{x,y}}{\pi(y)}-1\right|<\frac{1}{\mathrm{e}}\right\},\end{equation} to be the \textit{separation time} and the \textit{$\ell^\infty$-mixing time} respectively. 

	\section{Hitting and Cover Times}\label{S:regular}
	In this section we prove that the $\eps$-TBRW has the power to increase the probability of certain events. As a consequence of this result we obtain bounds on the cover and hitting times of the $\eps$-TBRW on a graph $G$ in terms of $n$, the extremal and average degrees, the relaxation time, and the hitting time of the SRW.
	
	The approach used to prove these results is, for a given graph $G$, to consider events which depend only on the trajectory of the walker (that is, the sequence of vertices visited) up to some fixed time $t$. We use a ``trajectory-tree'' to encode all possible trajectories. This then allows us to relate the probability of a given event in the $\eps$-TBRW to that for the SRW; the role of the technical lemma is to recursively bound the effects of an optimal strategy for the $\eps$-TBRW at each level of the tree. This section follows the conference version of this paper \cite{ITCSpaper} where the method was initially developed for the $\eps$-TBRW.  This method is flexible in the sense that it can be applied to other random processes with choice, in particular in \cite{POTC} we adapt this method to the choice random walk.

	Fix a vertex $u$, a non-negative integer $t$ and a set $S$ of trajectories of length $t$ (here the length is the number of steps taken). Write $p_{u,S}$ for the probability that running a SRW starting from $u$ for $t$ steps results in a member of $S$. Let $q_{u,S}(\eps)$ be the corresponding probability for the $\eps$-TBRW, which depends on the particular strategy used. It is important that neither the $\eps$-TBRW ($q$) nor the SRW ($p$) is lazy. We prove the following result relating $q_{u,S}(\eps)$ to $p_{u,S}$.
	
	\begin{theorem}\label{nonregboostnew}Let $G$ be a graph, $u\in V$, $t > 0$, $0\leq  \eps \leq  1$ and $S$ be a set of trajectories of length $t$ from $u$.  Then there exists a strategy for the $\eps$-TBRW 
		such that 
		\[
		q_{u,S}(\eps) \geq \left( p_{u,S} \right)^{1-\eps}.
		\]
	\end{theorem} 
	Here we typically think of $S$ encoding such events as ``the walker is in a set $W\subset V$ at time $t$'' or ``the walker has visited $v\in V$ by time $t$''; however, the result applies to any event measurable at time $t$. This theorem can be used to bound the cover time of the $\eps$-TBRW. 
	\begin{theorem}\label{trelbdd}
		For any graph $G$, and any $\eps\in(0,1)$,
		\[	\tetb(G)=\BO{\frac{\thit}{\eps}\cdot \log\left( \frac{\davg\cdot \trel \cdot \log  n}{\dmin} \right) }.\]
	\end{theorem}
	
	\cref{trelbdd} has the following consequence for expanders; a  sequence of graphs $(G_n)$ is a  \emph{sequence of expanders} if $\trel(G_n) = \BT{1}$. 
	
	\begin{corollary}\label{trelbddcor}For every sequence $(G_n)_{n\in \N}$ of $n$-vertex bounded degree expanders and any fixed $\eps >0$, we have  \[\tetb(G_n)=\BO{\frac{n}{\eps}\cdot \log \log n}.\]
	\end{corollary} 
We can also use \cref{nonregboostnew} to bound the hitting times of the $\eps$-BRW. 
	
	\begin{theorem}\label{trelhit}For any graph $G$, any $x,y \in V $ and any $\eps\in(0,1)$, we have
		\[\quad\Heb{x}{y}\leq 16 \cdot \pi(y)^{\eps-1}\cdot t_{\mathsf{mix}},\]
		where $\pi$ is the stationary distribution of the SRW;
		this bound also holds for return times. Additionally,	\[\tebh(G)\leq 120\cdot  \left(\frac{n\cdot \davg}{ d_{\mathsf{min}}}\right)^{1-\eps}\cdot \left( t_{\mathsf{mix}}\right)^{\frac{2+\eps}{3}}.\]
	\end{theorem}
 	We shall prove Theorem \ref{nonregboostnew} in \cref{gadget} after proving a key lemma in \cref{game}.  Theorem \ref{trelbdd} is an analogue for the $\eps$-TBRW of \cite[Theorem 6.1]{POTC}, and the derivation from \cref{nonregboostnew} follows that given in \cite[Section 6.1]{POTC} exactly. We include this for completeness in the Appendix. The statement of Theorem \ref{trelhit} is an improvement  over the analogous result in \cite[Theorem 6.2]{POTC} (the main improvement is replacing $\trel \log n$ with $ \tmix$). We give the proof of this result in \cref{sec:Hitproof} and note that the same proof will give the same improvement to \cite[Theorem 6.2]{POTC}.

	The main difference between the results here and those in \cite{POTC} is that each relies on an operator which describes the random walk process being studied. The operator used here is different to those introduced in \cite{POTC}, and as a result so is the strength of the boosting obtainable. This highlights the versatility of the technique used to prove \cref{nonregboostnew} in that it can be used to analyse several different random processes with non-deterministic interventions, such as the $\eps$-TBRW and the choice random walk (CRW) of \cite{POTC}.

	\subsection{The \texorpdfstring{$\eps$}{e}-Max/Average Operation}\label{game}
	In this subsection we shall introduce an operator which models the action of the $\eps$-TBRW. We shall then prove a bound on the output of the operator, which is used to show that the $\eps$-TBRW can boost probabilities indexed by paths.    
	
	For $0<\eps <1 $  define the $\eps$-max/average operator $\operatorname{MA}_{\eps}:[0,\infty)^m\to [0,\infty)$ by \begin{equation*}
	\operatorname{MA}_{\eps}\left(x_1,\dots , x_m \right)  = \eps \cdot \max_{1\leq i \leq m} x_i + \frac{1 - \eps }{m} \cdot  \sum_{i=1}^m x_i. 
	\end{equation*} 
	This can be seen as an average which is biased in favour of the largest element, indeed it is a convex combination between the largest element and the arithmetic mean.

	For $p\in \mathbb R\setminus\{0\}$, the $p$-power mean $M_p$ of non-negative reals $x_1,\ldots,x_m$ is defined by \[M_p(x_1,\ldots,x_m)=\left(\frac{x_1^p+\cdots+x_m^p}{m}\right)^{1/p},\]and \[M_{\infty}(x_1,\ldots,x_m)=\max\{x_1,\ldots,x_m\}=\lim_{p\to \infty}M_p(x_1,\ldots,x_m).\] Thus we can express the $\eps$-max/ave operator as $\operatorname{MA}_{\eps}(\cdot)=  (1-\eps)M_1(\cdot)+\eps \operatorname{M}_{\infty}(\cdot)$. We use a key lemma, \cref{anticonv}, which could be described as a multivariate anti-convexity inequality. 
	
	\begin{lemma}\label{anticonv}Let $0<\eps <1$, $m \geq 1$ and $\delta \leq  \eps /(1-\eps) $. Then for any $x_1,\dots, x_m \in [0,\infty)$, 
		\[M_{1+\delta}\left(x_1,\dots, x_m  \right)  \leq \operatorname{MA}_{\eps}\left(x_1,\dots, x_m  \right)  . \]
	\end{lemma}
	\begin{proof}	We begin by establishing the following claim.
		\begin{claim}Let $\eta\in(0,1)$, and suppose $a,b,c\in\mathbb R^+$ with $c=(1-\eta)a+\eta b$. Then
			\begin{equation}\label{mean-inequality}M_c\leq M_a^{(1-\eta)a/c}M_b^{\eta b/c}.\end{equation}
		\end{claim}
		\begin{poc}H\"older's inequality states for positive reals $y_1,\ldots,y_m$ and $z_1,\ldots,z_m$ that
			\[y_1z_1+\cdots+y_mz_m\leq\bigl(y_1^p+\cdots+y_m^p\bigr)^{1/p}\bigl(z_1^q+\cdots+z_m^q\bigr)^{1/q},\]
			where $p,q\geq 1$ satisfy $1/p+1/q=1$.
			The desired result follows by setting $y_i=x_i^{(1-\eta)a}$, $z_i=x_i^{\eta b}$, $p=1/(1-\eta)$, $q=1/\eta$, dividing both sides by $m$ and then taking $c$\textsuperscript{th} roots.
		\end{poc}
		Applying \eqref{mean-inequality}, we have for any $k>\delta$ that
		\begin{align*}
		M_{1+\delta}&\leq M_1^{\frac{1-\delta/k}{1+\delta}}M_{k+1}^{\frac{(k+1)\delta/k}{1+\delta}}\\
		&\leq\frac{1-\delta/k}{1+\delta}M_1+\frac{(k+1)\delta/k}{1+\delta}M_{\infty},
		\end{align*}
		using the weighted AM-GM inequality and the fact that $M_p\leq M_{\infty}$ for any $p$. Taking limits as $k\to\infty$, noting that $\eps\geq\delta/(1+\delta)$, gives the required inequality.
	\end{proof}

	\begin{remark}The dependence of $\delta$ on $\eps$ given in \cref{anticonv} is best possible. This can be seen by setting $x_1=0$ and $x_i=1$ for $2\leq i\leq m$, and letting $m$ tend to $\infty$.
	\end{remark}

	\subsection{The Trajectory-Tree for Graphs}\label{gadget}
	In this section we show how the ``trajectory-tree'' can be used to prove \cref{nonregboostnew}. This tree encodes walks of length at most $t$ from $u$ in a rooted graph $(G,u)$ by vertices of an arborescence $(\mathcal{T}_t,\mathbf{r})$, i.e.\ a tree with all edges oriented away from the root $\mathbf{r}$. Here we use bold characters to denote trajectories, and $\mathbf r$ will be the length-$0$ trajectory consisting of the single vertex $u$. The tree $\mathcal T_t$ consists of one node for each trajectory of length $i\leq t$ starting at $u$, and has an edge from $\mathbf{x}$ to $\mathbf{y}$ if $\mathbf{x}$ may be obtained from $\mathbf{y}$ by deleting the final vertex; we refer to such $\mathbf{y}$ as `offspring' of $\mathbf{x}$.  
	
	The proof of \cref{nonregboostnew} will follow the corresponding proof in \cite{POTC} closely, but we give a full proof here in order to clarify the role played by the $\eps$-max/average operator.
We write $d^+(\mathbf{x})$ for the number of offspring in $\mathcal T_t$ of $\mathbf x$, and $\Gamma^+(\mathbf{x})$ for the set of offspring of $\mathbf x$. Denote the length of the walk $\mathbf x$ by $\abs{\mathbf{x}}$. We shall extend our notation $p_{u,S}$ and $q_{u,S}(\eps)$ to $p_{\mathbf x,S}$ and $q_{\mathbf x,S}(\eps)$, defined to be the probabilities that extending $\mathbf x$ to a trajectory of length $t$, using the laws of the SRW and $\eps$-TBRW respectively, results in an element of $S$. Additionally, let $W_u(k):=\bigcup_{i=0}^k\{X_i \}$ be the trajectory of a simple random walk $X_t$ on $G$ up to time $k$, with $X_0=u$.

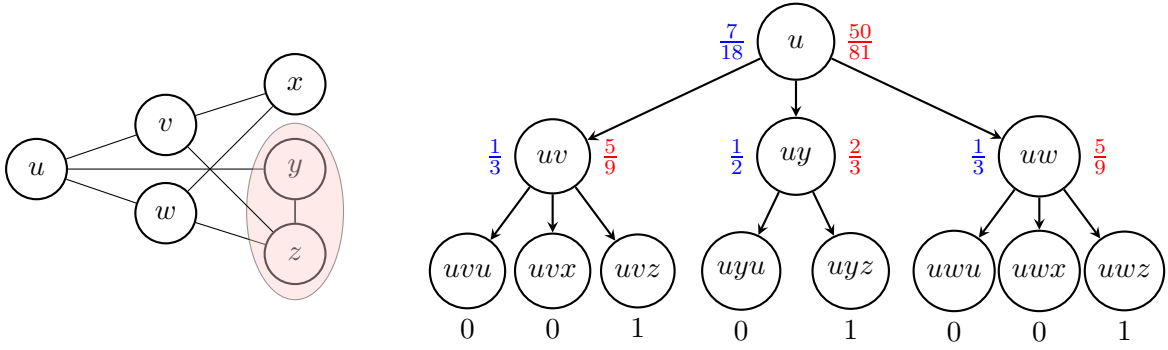
\begin{figure}
	\begin{subfigure}{.35\textwidth}
		\begin{tikzpicture}[xscale=.85,yscale=0.9,knoten/.style={thick,circle,draw=black,minimum size=.8cm,fill=white},wknoten/.style={thick,circle,draw=black,minimum size=.6cm,fill=white},edge/.style={black},dedge/.style={thick,black,-stealth}]
		\node[knoten] (u) at (2,2) {$u$};
		\node[knoten] (v) at (4,2.65) {$v$};
		\node[knoten] (w) at (4,1.35) {$w$};
		\node[knoten] (x) at (6,3.25) {$x$};
		\node[knoten] (y) at (6,2) {$y$};
		\node[knoten] (z) at (6,0.75) {$z$};
		\draw[edge] (u) to (v);
		\draw[edge] (u) to (y);
		\draw[edge] (w) to (x);
		\draw[edge] (w) to (z);
		\draw[fill=red!20,opacity=0.4] (6,1.375) ellipse (0.75cm and 1.3cm);
		\draw[edge] (v) to (x);
		\draw[edge] (z) to (y);
		\draw[edge] (u) to (w);
		\draw[edge] (v) to (z);
		\end{tikzpicture}
	\end{subfigure}%
	\begin{subfigure}{.65\textwidth}
		\begin{tikzpicture}[xscale=0.8,yscale=0.76,knoten/.style={thick,circle,draw=black,minimum size=.6cm,fill=white},wknoten/.style={thick,circle,draw=black,minimum size=.6cm,fill=white},edge/.style={black},dedge/.style={thick,black,-stealth}]
		\node[knoten] (u1) at (2,6) [label=left:$\textcolor{blue}{\frac{7}{18}}$,label=right:$\textcolor{red}{\frac{50}{81}}$]{$\phantom{x}u\phantom{x}$};
		\node[knoten] (v2) at (-2,4)[label=left:$\textcolor{blue}{\frac{1}{3}}$,label=right:$\textcolor{red}{\frac{5}{9}}$] {$\;uv\;$};
		\node[knoten] (y2) at (2,4)[label=left:$\textcolor{blue}{\frac{1}{2}}$,label=right:$\textcolor{red}{\frac{2}{3}}$] {$\;uy\;$};
		\node[knoten] (w2) at (6,4)[label=left:$\textcolor{blue}{\frac{1}{3}}$,label=right:$\textcolor{red}{\frac{5}{9}}$] {$\;uw\;$};
		
		\node[wknoten] (u3) at (-3.4,2)[label=below:$0$] {$uvu$};
		\node[wknoten] (x3) at (-2,2)[label=below:$0$] {$uvx$};
		\node[wknoten] (z3) at (-0.6,2)[label=below:$1$] {$uvz$};
		
		\node[wknoten] (u31) at (1.1,2)[label=below:$0$] {$uyu$};
		\node[wknoten] (z31) at (2.9,2)[label=below:$1$] {$uyz$};
		
		\node[wknoten] (u32) at (4.6,2) [label=below:$0$]{$uwu$};
		\node[wknoten] (x32) at (6,2) [label=below:$0$]{$uwx$};
		\node[wknoten] (z32) at (7.4,2) [label=below:$1$]{$uwz$};

		\draw[dedge] (u1) to (v2);
		\draw[dedge] (u1) to (y2);
		\draw[dedge] (u1) to (w2);
		
		\draw[dedge] (v2) to (u3);
		\draw[dedge] (v2) to (x3);
		\draw[dedge] (v2) to (z3);
		
		\draw[dedge] (y2) to (u31);
		\draw[dedge] (y2) to (z31);
		
		\draw[dedge] (w2) to (u32);
		\draw[dedge] (w2) to (x32);
		\draw[dedge] (w2) to (z32);
		
		\end{tikzpicture}
	\end{subfigure}
	\caption{Illustration of a (non-lazy) walk on a non-regular graph starting from $u$ with the objective of being at $\{y,z\}$ at step $t=2$. The probabilities of achieving this are given in blue (left) for the SRW and in red (right) for the $\frac{1}{3}$-TBRW.}
	
\end{figure}

	\begin{proof}[Proof of \cref{nonregboostnew}]For convenience we shall suppress the notational dependence of $q_{\mathbf{x},S}(\eps)$ on $\eps$. To each node $\mathbf{x}$ of the trajectory-tree $\ct_t$ we assign the value $q_{\mathbf{x},S} $ under the the $\eps$-TB strategy of biasing towards a neighbour in $G$ which extends to a walk $\mathbf{y}\in \Gamma^+(\mathbf{x})$ maximising $q_{\mathbf{y},S}$. This is well defined because both the strategy and the values $q_{\mathbf{x},S}$ can be computed in a ``bottom up'' fashion starting at the leaves, where if $\mathbf{x} \in V(\ct_t)$ is a leaf then $q_{\mathbf{x},S} $ is $1$ if $\mathbf x\in S$ and $0$ otherwise. 
		
		Suppose $\mathbf{x}$ is not a leaf. Then with probability $1-\eps$ we choose the next step of the walk uniformly at random in which case the probability of reaching $S$ from $\mathbf{x}$ is just the average of $q_{\mathbf{y},S}$ over the offspring $\mathbf{y}$ of $\mathbf{x}$, otherwise we choose a maximal $q_{\mathbf{y},S}$. Thus the value of $\mathbf{x}$ is given by the $\eps$-max/average of its offspring, that is \begin{equation}\label{qqqq}q_{\mathbf{x},S} =\operatorname{MA}_{\eps}\left( \left(q_{\mathbf{y},S}\right)_{\mathbf{y}\in \Gamma^+(\mathbf{x})} \right).\end{equation} 
		
		We define the following potential function $\Phi^{(i)}$ on the $i^{th}$ generation of the  trajectory-tree $\ct$: \begin{equation}\label{Phi}\Phi^{(i)}= \sum\limits_{\abs{\mathbf{x}}=i}q_{\mathbf{x},S}^{1+ \delta}\cdot \Pr{W_u(i) = \mathbf{x}}.   \end{equation}
		Notice that if $\mathbf{x}\mathbf{y}\in E(\ct_t)$ then \[\Pr{W_u(\abs{\mathbf{y}}) = \mathbf{y}} = \Pr{W_u(\abs{\mathbf{x}}) = \mathbf{x}}/d^+(\mathbf{x}) .\] Also since each $\mathbf{y}$ with $\abs{\mathbf{y}}=i$ has exactly one parent $\mathbf{x}$ with $\abs{\mathbf{x}}=i-1$ we can write 
		\begin{equation}\label{PPhi}\Phi^{(i)} = \sum\limits_{\abs{\mathbf{x}}=i-1}\sum_{\mathbf{y} \in \Gamma^+(\mathbf{x})}q_{\mathbf{y},S}^{1+ \delta}\cdot \frac{\Pr{W_u(i-1) = \mathbf{x}}}{d^+(\mathbf{x})  }.\end{equation} We now show that $\Phi^{(i)} $ is non-increasing in $i$. By combining \eqref{Phi} and \eqref{PPhi} we can see that the difference $\Phi^{(i-1)}-\Phi^{(i)}$ is given by  
		\begin{align*}
		&\sum\limits_{\abs{\mathbf{x}}=i-1} \left(q_{\mathbf{x},S}^{1+ \delta}-\frac{1}{d^+(\mathbf{x})}\sum_{\mathbf{y} \in \Gamma^+(\mathbf{x})}q_{\mathbf{y},S}^{1+ \delta} \right) \Pr{W_u(i-1) = \mathbf{x}}.
		\end{align*}Recalling \eqref{qqqq}, to establish $\Phi^{(i-1)}-\Phi^{(i)}\geq 0$ it is sufficient to show the following inequality holds whenever $\mathbf{x}$ is not a leaf: 
		\[ \operatorname{MA}_{\eps}\left(\left( q_{\mathbf{y},S}\right)_{\mathbf{y} \in \Gamma^+(\mathbf{x})} \right)^{1+ \delta} \geq \frac{1}{d^+(\mathbf{x} ) }\sum_{\mathbf{y}\in \Gamma^+(\mathbf{x})}q_{\mathbf{y},S}^{1+ \delta}.\]
		By taking $(1+\delta)$\textsuperscript{th} roots this inequality holds for any $\delta \leq  \eps /(1-\eps) $ by \cref{anticonv}, and thus for $ \delta$ in this range $\Phi^{(i)} $ is non-increasing in $i$. 
		
		Observe $\Phi^{(0)} = q_{u,S}^{1+\delta}$. Also if $\abs{\mathbf{x}}=t$ then $q_{\mathbf{x},S}=1 $ if $\mathbf{x} \in S$ and $0$ otherwise, it follows that  
		\[\Phi^{(t)} = \sum_{\abs{\mathbf{x}}=t}q_{\mathbf{x},S}^{1+\delta}\cdot \Pr{W_u(t) = \mathbf{x}}=\sum_{\abs{\mathbf{x}}=t}\mathbf{1}_{\mathbf{x}\in S}\cdot \Pr{W_u(t) = \mathbf{x}} = p_{u,S} .\] Thus since $\Phi^{(t)}$ is non-increasing $q_{u,S}^{1+\delta} = \Phi^{(0)}\geq \Phi^{(t)} = p_{u,S}  $. The result for the $\eps$-TBRW follows by taking $\delta = \eps /(1-\eps) $.   
	\end{proof}
	
	\subsection{Proof of Theorem \ref{trelhit}}\label{sec:Hitproof}
	We now prove Theorem \ref{trelhit}. The idea of the proof is to use Theorem \ref{nonregboostnew} to boost the probability that a random walk hits a vertex within $\Theta(\tmix)$ steps. 
	\begin{proof}Observe that for any non-negative integer random variable $Z$ the following holds\begin{equation}\label{eq:posrvtrick} \Pr{Z\geq 1} = \frac{\Ex{Z}}{\Ex{Z\mid Z\geq 1}}.\end{equation} Let $ N_y(T) = |\{t\leq T: \tilde{X}_t = y\}|$ be the number of visits to $y\in V$ up to time $T\geq 0$ by the lazy random walk $\tilde{X}_t$ on $G$. We shall now apply \eqref{eq:posrvtrick} to $N_y(T)$ for a suitable $T$.
		
		Recall the definition $t_{\mathsf{mix}}:=t_{\mathsf{mix}}(1/4)$ of the total variation mixing time. It follows that with probability $3/4$ we can couple a lazy random walk $\tilde{X}_t$ from any start vertex with a stationary walk by time $t_{\mathsf{mix}}$.  Then, for any $x,y\in V$ we have \begin{equation}\label{eq:bddonvisits}\Exu{x}{N_y(2 t_{\mathsf{mix}})}\geq \frac{3}{4}\cdot \Exu{\pi}{N_y(t_{\mathsf{mix}})}\geq \frac{3\pi (y)t_{\mathsf{mix}}}{4}.\end{equation} Now, if $N_y(T)\geq 1$ then $X_t$ first visited $y$ at some random time $0\leq s\leq T$. Taking $s=0$ gives   \begin{equation}\label{eq:bddonreturns}\Exu{x}{N_y(2t_{\mathsf{mix}})\mid N_y(2t_{\mathsf{mix}})\geq 1} \leq \sum_{t=0}^{2\cdot t_{\mathsf{mix}}}p_{y,y}^{(t)}\leq 2\cdot t_{\mathsf{mix}}+1\leq  3\cdot t_{\mathsf{mix}}.\end{equation} If we apply \eqref{eq:posrvtrick} to $N(y,X)$, then it follows from \eqref{eq:bddonvisits} and \eqref{eq:bddonreturns} that for any $x,y\in V$ we have 
		\begin{equation}\label{eq:hitprobbdd}\Pru{x}{ \tau_{y}\leq T}\geq \frac{\Exu{x}{N_y(2t_{\mathsf{mix}})} }{\Exu{x}{N_y(2t_{\mathsf{mix}}) \mid N_y(2t_{\mathsf{mix}})\geq 1}} \geq \frac{3\pi (y)t_{\mathsf{mix}}}{8}\cdot \frac{1}{3\cdot t_{\mathsf{mix}}} = \frac{\pi (y)}{8}. \end{equation}  By the natural coupling between trajectories of the simple and lazy random walks (adding in the lazy steps) it follows that \eqref{eq:hitprobbdd} also holds for the simple random walk. 	
		
		Now, applying Theorem \ref{nonregboostnew} to \eqref{eq:hitprobbdd} shows that for any $x,y\in V$ there exists a strategy for the $\eps$-TBRW to hit $y$ within $2\cdot t_{\mathsf{mix}}$ steps which has success probability at least $(\pi(y)/8)^{1-\eps}$. Thus if we run this strategy for $2\cdot t_{\mathsf{mix}}$ steps then repeat if necessary, we see that the expected time for the $\eps$-TBRW to hit $ y$ from $x$ is at most $2\cdot t_{\mathsf{mix}}/(\pi(y)/8)^{1-\eps} \leq 16\pi(y)^{\eps-1}\cdot t_{\mathsf{mix}} $. Since there exists an optimal strategy for hitting any vertex which is independent of time \cite[Theorem 5]{ABKLPbias} we conclude that this bound also holds for the $\eps$-BRW.

		To prove the second bound we shall get a different bound on returns (replacing \eqref{eq:bddonreturns}) which is independent of $y$. By \cite[Lemma 1]{oliveira2018random} and \cite[Lemma 2]{oliveira2018random}, for any $T\geq 0$ and $y\in V$, we have 
		\begin{equation}\label{eq:refinedret1}
		\sum_{t=0}^T p_{y,y}^{(t)} \leq \frac{e }{e-1}\sum_{t=0}^{\trel} p_{y,y}^{(t)} + T\cdot \pi(y) \leq \frac{e }{e-1}\cdot 6\pi(y)\frac{n\davg}{\dmin}\sqrt{\trel+1}  + T\cdot\pi(y). 
		\end{equation} Now since $\trel \leq \tmix\leq 2\thit +1 \leq 2n^3 + 1 \leq 3n^3$ by \cite[(10.24)]{levin2009markov} and \cite[(6.14)]{aldousfill} we have \begin{equation}\label{eq:refinedret2}n\sqrt{\trel +1 } + \tmix \leq 2n\sqrt{\tmix} + \tmix \leq 2n(\tmix)^{2/3} + \tmix\leq 4n(\tmix)^{2/3}.\end{equation} Thus by \eqref{eq:refinedret1} and \eqref{eq:refinedret2} and since $12e/(e-1)< 19$ we have \begin{equation}\label{eq:refinedret3}  \sum_{t=0}^{2\tmix} p_{y,y}^{(t)} \leq \pi(y)\left( \frac{6e}{e-1}\frac{\davg}{\dmin}\cdot n\sqrt{\trel+1}  + 2\tmix\right)\leq 21\pi(y)\frac{n\davg}{\dmin}\left(\tmix\right)^{2/3}.  \end{equation}Now, using the bound \eqref{eq:refinedret3} on $\Exu{x}{N_y(2t_{\mathsf{mix}})\mid N_y(2t_{\mathsf{mix}})\geq 1}$ instead of \eqref{eq:bddonreturns} in \eqref{eq:hitprobbdd} gives us \begin{equation*}\Pru{x}{ \tau_{y}\leq T}\geq  \frac{3\pi (y)t_{\mathsf{mix}}}{8}\cdot \frac{1}{21\pi(y)\frac{n\davg}{\dmin}\left(\tmix\right)^{2/3}} \geq \frac{\dmin (\tmix)^{1/3}}{60 n\davg }. \end{equation*}Now, by the same steps as before there is a strategy for the $\eps$-BRW to hit any $y$ from any $x$ in time at most $\left(\frac{60 n\davg }{\dmin (\tmix)^{1/3}}\right)^{1-\eps}\cdot 2\tmix \leq 120 \left(  n\davg /\dmin\right)^{1-\eps}(\tmix)^{\frac{2+\eps}{3}}  $ as claimed.\end{proof}
 
	\section{Increasing Stationary Probabilities}\label{AzarConjSec}
	In this section we shall consider the problem of how much an unchanging strategy can affect the stationary probabilities in a graph. Azar et al.\ studied this question and made an appealing conjecture. Our result on the hitting times of the $\eps$-BRW allows us to make progress towards this conjecture. We also derive some more general bounds on stationary probabilities for classes of Markov chains which include certain regimes for the $\eps$-BRW, and tackle the question of when the stationary probability of a vertex can be changed by more than a constant factor.   
	\subsection{A Conjecture of Azar et al.}
	Azar, Broder, Karlin, Linial and Phillips make the following conjecture for the $\eps$-BRW \cite[Conjecture 1]{ABKLPbias}. Their motivation was that a corresponding bound holds for the related process studied by Ben-Or and Linial \cite{ben1987collective}.
	\begin{conjecture}[ABKLP Conjecture]\label{abklp}
		In any graph, a controller can increase the stationary probability of any vertex from $p$ (for the SRW) to $p^{1-\eps}$.
	\end{conjecture}
	This conjecture becomes particularly attractive in the context of \cref{nonregboostnew}, which implies that in the $\eps$-TBRW a controller may increase the probability of being at any given vertex at time $t$ from $p_t$ to $p_t^{1-\eps}$, where for non-bipartite graphs we have $p_t\to p$. However, a crucial point is that the strategy guaranteed by \cref{nonregboostnew} depends on $t$, and so we cannot necessarily achieve this boosting uniformly over $t$, or by using only the $\eps$-BRW.
	
	In fact, the conjecture fails for the graph $K_2$, as no strategy for the $\eps$-BRW can increase the stationary probability over that of a simple random walk. This motivates weakening the conjecture by replacing $ p^{1-\eps}$ by $p^{1- \eps + o_n(1)} $; however this fails for the star on $n$ vertices, and non-bipartite counterexamples may be obtained by adding a small number of extra edges to the star. In each of these counterexamples there is a vertex with constant stationary probability, and for large graphs this can only happen if there is a large degree discrepancy. 
	We believe the following should hold.
	\begin{conjecture}\label{reformulated}
		In any graph a controller can increase the stationary probability of any vertex from $p$ to $p^{1-\eps+\delta} $, where $\delta \rightarrow 0$ as $p \rightarrow 0$. 
	\end{conjecture}
	
	Azar et al.\ prove a weaker bound of $p^{1-\mathcal{O}(\eps)}$ for bounded-degree regular graphs. As a corollary of \cref{trelhit} we confirm \cref{reformulated} for any graph where $t_{\mathsf{mix}}$ is sub-polynomial in $n$. Our techniques are different to those of Azar et al.\ and allow us to cover a larger class of graphs, including dense graphs as well as sparse ones. In addition, for graphs where $\dmax/\davg$ and $t_{\mathsf{mix}}$ are both sub-polynomial our result achieves the same exponent (up to lower order terms) as the conjectured bound. 
	\begin{theorem}\label{azarconj}In any graph a controller can increase the stationary probability of any vertex from $p$ to $p^{1-\eps+\delta}$, where $\delta:=\delta_G=\ln\left( 16\cdot \tmix\right)/\abs{\ln p}$.
	\end{theorem}
	\begin{proof}By \cref{trelhit} for each vertex $v$ there exists a strategy so that the return time to $v$ is at most $16\cdot \pi(v)^{\eps-1}\cdot \tmix $. Let $q$ denote the stationary probability of $v$ for this $\eps $-B walk. Then as stationary probability is equal to the reciprocal of the return time by \cite[Prop.\ 1.14]{levin2009markov} we have $q\geq \pi(v)^{1-\eps}/(16\tmix)  $, for the simple random walk $p=\pi(y)$. For $\delta = \ln\left( 16\tmix\right)/\abs{\ln\pi(y)}$ we have \begin{align*}q/p^{1-\eps+\delta } &\geq \frac{\pi(v)^{1-\eps}}{16\tmix}\cdot\frac{\pi(y)^{-\delta}}{\pi(v)^{1-\eps}} = \frac{\exp\left(-\ln\pi(y) \cdot \frac{\ln \left(16\cdot \tmix\right)}{\abs{\ln\pi(y)}}  \right) }{16\tmix}  = 1.\qedhere\end{align*}
	\end{proof} 
	The dependence of $\delta$ on $\abs{\ln p}$ in \cref{azarconj} imposes the condition that any vertex we wish to boost must have sub-polynomial degree. This condition is tight in some sense as no stationary probability bounded from below can be boosted by more than a constant factor. In \cref{s:polyboost} we prove a weaker boosting effect which holds in any sublinear polynomial-degree regime.
	
	In the context of $d$-regular graphs, Azar et al.\ state, \begin{displayquote}[\cite{ABKLPbias}][]The interesting situation is when $\eps$ is not substantially larger than $1/d$; otherwise, the process is dominated by the controller's strategy. \end{displayquote}
	
	Note that for $d$-regular graphs with $d=\omega(\log n)$ the conjectured boost from $p$ to $p^{1-\eps}$ does not change the stationary probabilities by more than a constant factor in this regime. For this reason we shall focus on the following question for $d$-regular graphs. \begin{quest}\label{const-fac}When can we boost the stationary probability by more than a constant factor in the $\eps$-BRW with $\eps=\BT{1/d}$?\end{quest}
	As noted when $d=\omega(\log n)$ such a boost is stronger than for the AKBLP conjecture and we think \cref{const-fac} is quite natural. 
	
	We will consider not only regular graphs but also \textit{almost-regular} ones, that is, graphs in which degrees differ by at most a constant factor. An interesting class is the almost-regular graphs of linear degree; we say that a graph is \textit{everywhere dense} if it has minimum degree $\Omega(n)$. We consider \cref{const-fac} for such graphs in \cref{const-boost}. In particular, we show that the answer to \cref{const-fac} is negative for everywhere dense graphs.
	This is essentially best possible, since we show that the corresponding result does not hold for $n^{\alpha}$-regular graphs for any $\alpha<1$. However, it does hold for almost every almost-regular graph in this regime.
	
	\subsection{Boosting in the polynomial degree regime}\label{s:polyboost}
	In this section we prove the following boosting result for graphs whose degree is bounded by a polynomial function of $n$.
	\begin{corollary}\label{cor-small-poly}Let $G$ be any graph satisfying $\dmax\leq n^{\alpha}$ for some $\alpha\in(0,1)$. Then a controller for the $\eps$-BRW can increase the stationary probability of any vertex from $p$ to $p^{1-c_\alpha\eps/\ln \dmax}$ for some $c_\alpha>0$.\end{corollary}
	
	Let $G=(V,E)$ be any connected, undirected graph with degree bound $d \leq C$. We will associate to every edge a positive weight given by the function $w:E \rightarrow \mathbb{R}^{+}$. We consider a random walk that picks an incident edge with probability proportional to its weight. Recall that the stationary distribution of this walk is given by $\pi(x)=  \sum_{y \sim x} w(y,x)/(2W)$,
	where $W := \sum_{ \{r,s\} \in E(G) } w(r,s)$ is the total sum of weights assigned. 
	
	Fix a vertex $u \in V$ and let $-1< a < \infty$. We consider the weight function given by
	\begin{equation}\label{weightedwalk}
	w(r,s) = \left( 1 + a \right)^{ \max\{d(u,r),d(u,s) \}},
	\end{equation}
	where $d(\cdot,\cdot)$ is the graph distance.
	Note that this particular weight function satisfies the following property:
	\begin{equation}\label{ratiocond}
	\forall x,y,z \colon \{x,y\},\{x,z\} \in E(G) \colon \frac{w(x,y)}{w(x,z)} \in\{1 + a,(1+a)^{-1},1\}.
	\end{equation}
	
	\begin{proposition}
		Let $-1< a< \infty$, and let $G$ be an edge-weighted graph whose weights satisfy \eqref{ratiocond}. Then, provided $\eps \geq-a$ if  $a\leq  0$ and $\eps \geq  a/(1+a) $ if $a>0$, the $\eps$-BRW can emulate the walk given by those weights.
	\end{proposition}
	\begin{proof}
		It suffices to prove that we may emulate a step of the walk from any given vertex $x$. If all edges meeting $x$ have the same weight, we simply ``bias'' towards the uniform distribution on neighbours of $x$. Otherwise $a\neq 0$, $d=d(x)\geq 2$ and there are exactly two weights, $w_1$ and $w_2$, incident to $x$, which satisfy $w_1=(1+a)w_2$. Suppose there are $k$ incident edges of weight $w_1$ and $d-k$ of weight $w_2$; clearly $1\leq k\leq d-1$.
		Now we need to construct a bias matrix $\mathbf B$ which will satisfy the walk probabilities given by \eqref{weightedwalk}. Note that if $w(xy)=w_1$ then $p_{x,y} = w_1/(kw_1 + (d-k)w_2) = (1+a)/(ak + d)$ and otherwise $p_{x,y}=1/(ak +d )$. 
		
		We first consider the case $a>0$, i.e.\ $w_1>w_2$. It is sufficient to assume $\eps=\frac{a}{1+a}$, since if it is larger we may use the $\eps$-BRW to emulate the $\frac{a}{1+a}$-BRW. 
		In this case set 
		\[\mathbf B_{x,z}=\begin{cases}\frac{da+2d-k}{dak+d^2}&\text{ if }w(xz)=w_1\\
		\frac{d-k}{dak+d^2}&\text{ if }w(x,z)=w_2.\end{cases}\]
		This gives $\sum_{z\sim x}\mathbf{B}_{x,z}=1$, all entries are positive and
		\[p_{x,z}=\frac{a}{1+a}\cdot\mathbf{B}_{x,z}+\frac{1}{1+a}\cdot\frac 1d=\begin{cases}\frac{a+1}{ka+d}&\text{ if }w(xz)=w_1\\
		\frac{1}{ka+d}&\text{ if }w(x,z)=w_2.\end{cases}\]
		The case $a<0$ may be reduced to the previous case by replacing $a$ with $a'=\frac{-a}{1+a}$, noting that $\eps\geq -a$ is equivalent to $\eps\geq\frac{a'}{1+a'}$.
	\end{proof}	
	\begin{theorem}\label{boostingp}Let $G$ be any graph such that $\dmax\geq 3$ and let $\eps >0$. Then a controller for the $\eps$-BRW can increase the stationary probability of any vertex from $p$ to $p^{1-\tilde{\eps}}$, where \[\tilde{\eps}=  \frac{\ln(1-\eps)\ln p}{\ln(\dmax -1)\ln n}>0.\]
	\end{theorem}
	\begin{proof}Consider a walk $\mathbf{Q}$ with weighting scheme $w(r,s) = \left( 1 - \eps \right)^{\max\{d(u,r),d(u,s) \}}$. Note there are at most $\dmax(\dmax -1)^{i-1} $ vertices at distance exactly $i$ from $u$ (and also edges from vertices at distance $i-1$ to those at $i$). Thus, writing $W$ for the total weight of the graph, for any $r$, 
		\begin{align*}W&\leq  \sum_{i=1}^r \dmax (\dmax -1)^{i-1}\cdot (1-\eps)^{i-1} + n \cdot \davg \cdot (1-\eps)^{r}\\ &\leq \left( 2(\dmax -1)^r + n\cdot \davg \right)\cdot (1-\eps)^{r}. \end{align*}  Thus if we let $ r =\lfloor\ln(n)/\ln( \dmax -1)\rfloor  $ then $W \leq \davg \cdot n^{1+\kappa }$, where $\kappa =\ln(1-\eps)/ \ln(\dmax -1) <0$. For any $u \in V$ it follows that $\pi_{\mathbf{Q}}(u) \geq d(u) /\davg \cdot n^{1+\kappa }  = n\cdot\pi(u) /n^{1+\kappa }$ and so for $\delta\geq 0$, 
		\[\frac{\pi_{\mathbf{Q}}(u)}{ \pi(u)^{1+\kappa + \delta } } \geq \frac{ n\cdot\pi(u) }{n^{1+\kappa }}  \cdot \frac{n^{1+\kappa + \delta }}{ (n\cdot \pi(u))^{1+\kappa + \delta } }= (n\cdot \pi(u))^{-\kappa - \delta } \cdot n^{\delta} \geq 1,\] where the final inequality holds by taking $\delta = \abs{\kappa \ln(n\pi(u) )}/\ln n  $. 
	\end{proof}
	\begin{proof}[Proof of \cref{cor-small-poly}]The statement holds for paths and cycles, and for graphs such that $\dmax \geq 3$ it follows from \cref{boostingp} since $-\ln(1-x) \geq x$ for any $x\leq 1$.
	\end{proof}

	\subsection{Boosting by more than a constant factor}\label{const-boost}
	In this section we show that in the case of an everywhere-dense graph, stationary probabilities for the $\eps$-BRW cannot exceed those for the SRW by more than a constant factor, giving a negative answer to \cref{const-fac}. In fact we show that this bound applies more generally to a class of (not necessarily reversible) Markov chains which resemble simple walks on everywhere-dense graphs. In contrast, we show that there exist regular graphs with polynomial degree arbitrarily close to linear for which the answer to \cref{const-fac} is positive. However, such graphs are rare: the answer is negative with high probability for a random graph with the same density, and hence for almost all almost-regular graphs in the polynomial regime.

Let $\mathbf{Q}=(q_{u,v})_{u,v\in V}$ be a transition matrix supported on $G$. For $c,C$ such that $0<c \leq C<\infty $ we say that the corresponding Markov chain is a \textit{$(c,C)$-simple walk} on $G$ if for every $uv \in E(G)$,  
\[\frac{c}{ d(u)} \leq q_{u,v} \leq \frac{C}{d(u)}.\]
\begin{proposition}\label{prop:dense}
For any graph $G$ with minimum degree $\dmin\geq \alpha \cdot n$ for some constant $\alpha>0$, any strategy $\mathbf{Q}$ for the $\eps$-BRW with $\eps\leq\beta/n$ satisfies $\pi_\mathbf{Q}(u)\leq(1+\beta)\alpha^{-2}\pi(u)$ for every $u\in V$.
\end{proposition}
\begin{proof}
Note that any strategy for the $\eps$-BRW satisfies 
\[\frac{1}{d(u)}(1-\eps)\leq q_{u,v}\leq\eps+\frac{1}{d(u)}(1-\eps),\]
and, since $\eps\leq\beta/n\leq\beta/d(u)$, this is a $((1-\eps),(1+\beta))$-simple walk on $G$. Noting that 
\[\pi(u)=\frac{d(u)}{\sum_{v\in V}d(v)}\geq\frac{\alpha}{n},\]
it is sufficient to verify that for any $(c,C)$-simple walk on $G$, the stationary probability $\pi_{\mathbf{Q}}$ satisfies $\pi_{\mathbf{Q}}(u)\leq C/(\alpha n)$ for every $u\in V$. This is true since
\[\pi_{\mathbf{Q}}(u) = \sum_{v\in V}\pi_{\mathbf{Q}}(v)q_{v,u} \leq \frac{C}{\dmin }\sum_{v\in V}\pi_{\mathbf{Q}}(v)\leq\frac{C}{\alpha n}.\qedhere\] 
\end{proof}

We shall now give some bounds on the stationary distribution of $(c,C)$-simple walks based on variants of the mixing time. Recall the definitions of the $\ell^\infty$ mixing time $t_{\infty}$ and the separation time $\tsep$ from Section~\ref{formaldef} and note that throughout  we follow the convention that $1/0 =\infty$.

\begin{proposition}\label{pibdd} Let $G$ be a connected graph and ${\mathbf{Q}}$ be a $(c,C)$-simple walk on $G$. Let $\tau_1 =\min\left\{\tsep, \;3\log(n) / |\log \lambda_*|   \right\}$ and $\tau_2 =\min\left\{t_{\infty}, \;3\log(n) / |\log \lambda_*|   \right\}$. Then
	
		\[
	\phantom{\qquad \text{for all }x \in V}\frac{c^{\tau_1}}{2}\cdot \pi(x)  \leq \pi_{\mathbf{Q}}(x) \leq 2 C^{\tau_2} \cdot \pi(x)\qquad \text{for all }x \in V.
	\]

\end{proposition}
\begin{proof} Let $\tilde{\mathbf{P}}$ and $\tilde{\mathbf{Q}}$ be the LRW and lazy $(c,C)$-simple walk on $G$ respectively and observe that for any $u,v \in E$ we have $\tilde{q}_{u,v} = (1+{q}_{u,v})/2\geq (1+c{p}_{u,v})/2\geq c(1+{p}_{u,v})/2\geq c\tilde{p}_{u,v}$ since $c\leq 1$. Thus $\tilde{q}^t_{y,x}\geq c^t\cdot \tilde{p}^t_{y,x} $ for any $t\geq 1$ and $x,y\in V$. Recall the definition of the separation distance $\tsep $ from  $\tilde{p}_{x,y}^{(\tsep)}\geq \frac{e-1}{e}\cdot \pi(y)  $  for any $x,y \in V$. Thus for any $x\in V$ we have  
	\begin{equation}\label{eq:statbounds}\begin{aligned}
	\pi_{\mathbf{Q}}(x) &= \sum_{y \in V} \pi_{\mathbf{Q}}(y) \cdot \tilde{q}^{(\tsep)}_{y,x} \\
	&\geq c^{\tsep} \cdot \sum_{y \in V} \pi_{\mathbf{Q}}(y) \cdot \tilde{p}^{(\tsep)}_{y,x} \\
	&\geq c^{\tsep} \cdot  \sum_{y \in V} \pi_{\mathbf{Q}}(y)\cdot \frac{e-1}{e}\pi(x) \\
	&= c^{\tsep} \cdot\frac{e-1}{e}\cdot  \pi(x).
	\end{aligned}\end{equation}
	For the upper bound recall the definition of the $\ell^\infty$-mixing time $t_{\infty}<\infty$ from \eqref{eq:sep} and observe that $\tilde{p}_{x,y}^{(t_{\infty})}\leq \frac{e+1}{e}\cdot \pi(y)$ for any $x,y\in V$. Thus by similar steps as \eqref{eq:statbounds} we have \[ \pi_{\mathbf{Q}}(x)\leq C^{t_{\infty}}\cdot \frac{e+1}{e}\cdot \pi(x).\] If the graph $G$ is aperiodic then $\lambda_*<1$ for the SRW $\mathbf{P}$. In this case we recall the following inequality, $\left|p_{x,y}^{(t)}/\pi(y)-1\right|\leq  \lambda_*^t/\min_{x\in V}\pi(x)  $ for any $t\geq 1$ and $x,y\in V$ by \cite[(12.11)]{levin2009markov}. Thus, since $\min_{x\in V}\pi(x)\geq 1/n^2$, if we take $t= 3\log(n)/|\log \lambda_*|$ then we have  \[\pi(y)/2\leq \pi(y)\left(1-1/n\right)\leq p_{x,y}^{(t)}\leq \pi(y)\left(1+1/n\right)\leq 2\pi(y),\] as we can assume $n\geq 2$ or else the result holds vacuously. Consequently, again similarly to \eqref{eq:statbounds}, we have $c^{t}\cdot \pi(x)/2 \leq \pi_{\mathbf{Q}}(x)\leq C^{t}\cdot 2 \pi(x) $. The result follows by taking the maximum of the first two bounds with $t$ and observing that $\frac{e-1}{e}\geq 1/2 $ and $\frac{e+1}{e}\leq 2$.
\end{proof}

	Now we show that for the Erd\H{o}s--R\'{e}nyi random graph in the polynomial average degree regime the answer to \cref{const-fac} is negative w.h.p.

	\begin{proposition}\label{GnpNoBoost}Let $0\leq \beta<\infty $ be a fixed real and $\mathcal{G}\overset{d}{\sim}\mathcal{G}(n,p)$ where $np\sim n^\alpha$ for some fixed real $0<\alpha\leq 1$. Then w.h.p.\ for every vertex $u$ the controller of $(\beta /np)$-BRW can only increase the stationary probability of $u$ from $\pi(u)$ to at most $3 \left(1+\beta\right)^{6/\alpha}\cdot \pi(u) $.
	\end{proposition}
	\begin{proof} To begin, by the union and Chernoff bounds \cite[Cor. 4.6]{MitzUpfal} we have \[\Pr{\cup_{x\in V}\left\{|d(x) - np|> 3\sqrt{np \log n}\right\}}\leq  n\cdot 2\exp\left(-np\left(3\sqrt{\frac{\log n}{np }}\right)^2/3 \right) \leq \frac{1}{n^2} .\] Thus w.h.p., for any $(\beta/np)$-BRW strategy $Q$ we have  
		\[q_{x,y} \leq \frac{\beta}{np} + \frac{1-\beta/np}{d(x)} \leq  \frac{1 }{d(x)}\left(1+\beta +100\beta\cdot \sqrt{\frac{\log n}{np}}\right).\] Since also $q_{x,y} \geq (1-\eps)/d(x) = (1-\beta/np)/d(x)$, we see that for any fixed strategy, $Q$ is a $\left(1-\beta/np, \;  1+\beta+100\beta\cdot \sqrt{ (\log n)/np }\right)$-simple walk.
		
		For a graph $G$ let $\mathcal{L} = \mathbf{I}-\mathbf{D}^{-1/2}\mathbf{A}\mathbf{D}^{1/2} $ the \textit{normalised Laplacian}, where $\mathbf{D}$ is a diagonal matrix with $ d_{x,x} = d(x)$, $\mathbf{A}$ is the adjacency matrix, and $\mathbf{I}$ is the identity matrix. By \cite[Thm.\ 1.2]{CojaEigen} there exists some $c<\infty $ such that if $np \geq c\log n$ then w.h.p.\ we have \begin{equation}\label{eq:lap}1- (4+\lo{1})/\sqrt{np} \leq \lambda_2(\mathcal{L}(\mathcal{G}(n,p)))\leq \lambda_n(\mathcal{L}(\mathcal{G}(n,p)))\leq 1+ (4+\lo{1})/\sqrt{np}.\end{equation} Observe that since the diagonal matrix $\mathbf{D}$ is invertible, the matrices $\mathcal{L}$ and $\mathbf{D}^{-1/2}\mathcal{L}\mathbf{D}^{1/2} = \mathbf{I} - \mathbf{P}$ are similar (and thus have the same eigenvalues). Thus, by shifting the eigenvalues of $\mathcal{L}$ to correspond to the SRW $\mathbf{P}$, \cite[Thm.\ 1.2]{CojaEigen} implies that for $np\geq c\log n$ we have $\lambda_* \leq (4+\lo{1})/\sqrt{np}$ w.h.p.. Thus we have $|\log \lambda_*|\geq (\alpha/2) \log n  - \log4 -\lo{1}$ w.h.p. and consequently $3\log(n)/|\log \lambda_*| \leq (6/\alpha)( 1  + 2/\log n) $ w.h.p.\ for large $n$. Thus by Proposition \ref{pibdd} we have  \[\pi_{\mathbf{Q}}(x)\leq 2\left( 1+\beta+100\beta\cdot \sqrt{\frac{\log n}{n^\alpha}}\right)^{(6/\alpha)(1+2/\log n)}\pi(x) \leq 3 \left(1+\beta\right)^{6/\alpha}\pi(x) ,\]w.h.p.\ for suitably large $n$ when $0<\alpha\leq 1$ and $ 0\leq \beta<\infty$ are fixed, as claimed.
	  \end{proof}
			
	Finally, we give a general $d$-regular example with $d=\poly(n)$ for which we can answer \cref{const-fac} in the affirmative. These graphs have the largest possible diameter $\approx n/d$ and feature several bottlenecks. 
	\begin{proposition}\label{RegCycleBoost}
		Fix any $0<\alpha < 1$ and let $d= n^{\alpha}$, $\eps=\Theta(1/d)$. Then there exists a $d$-regular graph for which the stationary distribution of any given vertex can be boosted by the $\eps$-TB random walk from $1/n$ to $\Omega(1/n^{\alpha})$.
	\end{proposition}
	\begin{proof}Let $d= n^{\alpha}$ and $\ell =n^{1-\alpha}$ and consider the $(\ell,K_{d,d})$-ring pictured in \cref{fig:my-label}. The $(\ell,K_{d,d})$-ring has $ N=2\ell (d+1)$ vertices and is $d+1$-regular graph, thus in our case $N  \sim 2n$.  
		
		Let $x,u$ be the end points of one of the edges which connects two units, and $u_1,\dots, u_d$ be the vertices in the $K_{d,d}$ attached to $u$ (see \cref{fig:my-label}). Assuming that $x$ is closer to the target vertex we wish to boost, the $\eps$-BRW strategy is clear: we should prefer the walk at $u$ to visit $x$ and thus set $B_{u,x}=1$ and $B_{u,u_i}=0 $, for all $1\leq i\leq d$ where $B$ is the bias matrix. Now we see that \[\frac{w(u,x)}{w(u,u_i)} = \frac{\eps + (1-\eps)/(d+1)}{(1-\eps)/(d+1)} = 1 + \frac{\eps (d+1)}{(1-\eps)} = 1+ \Omega(1). \]We seek to bound the total weight $W$. If we sum from the target $v$, where we set the adjacent weights to $1$, then we see that the $i$\textsubscript{th}  $K_{d,d}$ away from $v$ must have weights that are at most $(1+ \Omega(1))^{-i}$, thus   
		\[ W\leq 2 \sum_{i=0}^{\ell} (1+\Omega(1) )^{-i}(d^2 + 2d +1) = \mathcal{O}(d^2). \] 
		Now we see a boosting under this $\eps$-TB boosting strategy from $1/N$ to $p'$ where \[p'\geq  d/\mathcal{O}(d^2) = \Omega (1/d ) = \Omega(N^{-\alpha}).\qedhere\] 
	\end{proof}

	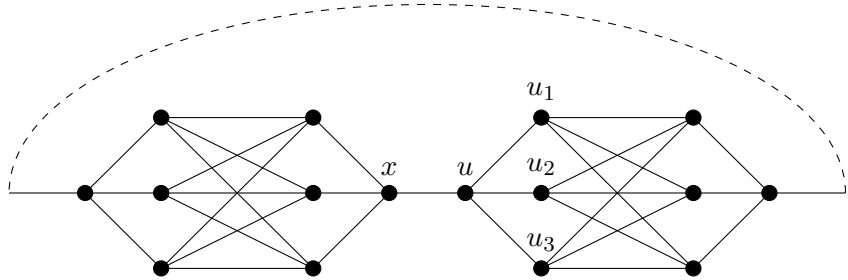
\begin{figure}[!htb]
		\center\begin{tikzpicture}
		\foreach \x in {0,4,5,9}
		\draw[fill] (\x,1) circle (.1);
		\foreach \x in {1,3,6,8}
		\foreach \y in {0,1,2}
		\draw[fill] (\x,\y) circle (.1);
		\foreach \x in {1,6}
		\foreach \y in {0,1,2}{%
			\draw (\x,\y) -- (\x-1,1);
			\draw (\x+2,\y) -- (\x+3,1);
			\foreach \z in {0,1,2}
			\draw (\x,\y) -- (\x+2,\z);}
		\foreach \x in {-1,4,9}
		\draw (\x,1) -- (\x+1,1);
		\draw (4,1.1) node[anchor=south]{$x$};
		\draw (5,1.1) node[anchor=south]{$u$};
		\foreach \x in {1,2,3}
		\draw (6,3.1-\x) node[anchor=south]{$u_{\x}$};
		\draw[dashed] (10,1) arc(0:180:5.5cm and 2.5cm);
		\end{tikzpicture}
		\caption{\label{fig:my-label}The $(\ell,K_{d,d})$-ring consists of $\ell$ complete bipartite graphs on $d$ vertices arranged in a cycle. The $(\ell,K_{3,3})$-ring, for some $\ell\geq 2$, is shown above.}
	\end{figure}

	\section{Computing Optimal Choice Strategies} \label{complexsec}
	In this section we focus on the following problem: given a graph $G$ and an objective, how can we compute a strategy for the $\eps$-TBRW which achieves the given objective in optimal expected time? Unless otherwise specified, this section considers walks on the more general class of strongly-connected directed graphs. A strategy consists of a family of controller bias matrices $\{\textbf{B}(\mathcal H_t)\}$, where $t\geq 0$ is the time and $\mathcal H_t$ is the history of the walk up to time $t$. Azar et al.\ \cite{ABKLPbias} considered the following computational problems:
	\begin{labeling}{$\mathtt{Hit}\left(G,v,S\right)$:}
		\item[$\mathtt{Stat}\left(G,w\right)$:] Find an $\eps$-bias strategy min/maximising $\sum_{v  \in V}w_v \cdot \pi_{v}$ for vertex weights $w_v\geq 0$.
		\item[$\mathtt{Hit}\left(G,v,S\right)$:] Find an $\eps$-bias strategy minimising $\sum_{v  \in V}\ell_v\cdot \Heb{v}{S}$ for a given $S\subseteq V(G)$,  $v\in V(G)$  and vertex weights $\ell_v\geq 0$.
	\end{labeling}
	Notice that for $\mathtt{Stat}$ to make sense we must fix an unchanging strategy and there exists an unchanging optimal strategy for $\mathtt{Hit} $, see \eqref{bias}. Azar et al.\ showed $\mathtt{Stat}$ and $\mathtt{Hit}$ are tractable. 
	\begin{theorem}[Theorems 6 \& 12 in \cite{ABKLPbias}]\label{azarpoly} Let $G$ be any connected directed graph, $v\in V(G)$ and $S\subseteq V(G)$. Then $\mathtt{Stat}\left(G,w\right)$ and $\mathtt{Hit}\left(G,v,S\right)$ can be solved in polynomial time.
	\end{theorem}We introduce the following computational problem not considered by Azar et al. 
	\begin{description}
		\item[$\mathtt{Cov}\left(G,v\right)$:] Find an $\eps$-TB strategy minimising  $\ETBcov{v}{G}$ for a given $v \in V(G)$.
	\end{description}
	Unlike for $\texttt{Stat}$ and $\texttt{Hit}$, an optimal strategy for $\texttt{Cov}$ on essentially any graph cannot be unchanging as it will need to adapt as some vertices become visited (consider the walk on a path started from the midpoint). \cref{Covunchanging} shows that there is an optimal strategy for $\mathtt{Cov}$ which is conditionally independent of time, in that no more information from $\mathcal H_t$ than the set of uncovered vertices is used. This fact means that an optimal strategy for $\mathtt{Cov}$ can be described using only finitely many bias matrices.
	
	Additionally one can show that, for undirected graphs, the $\eps$-TBRW exhibits the same dichotomy as the CRW studied in \cite{POTC}, by a simple adaptation of the proof of hardness in \cite{POTC}. That is while optimising $\mathtt{Hit}$ admits a polynomial-time algorithm, even computing an individual bias matrix $\textbf{B}(\mathcal H_t)$ from an optimal strategy for $\mathtt{Cov}$ is $\NP$-hard. We may view this as an on-line approach to solving $\mathtt{Cov}$, where we compute only the specific bias matrices needed as the random walk progresses; clearly this is an easier problem than precomputing an entire optimal strategy. Note that at most $n$ bias matrices will need to be computed in the course of any given walk, since an optimal bias matrix only depends on the uncovered set, which changes at most $n$ times; however, a full optimal strategy may require exponentially many such matrices.
	
	In fact we will prove $\PSPACE$-completeness for the (online) covering problem in the more general setting of directed graphs. Again we consider the on-line version of the problem, which represents computing a single row of the bias matrix. The input is a (directed) graph $G$, a current vertex $u$, and a visited set $X$ containing $u$. We require $G$ to be strongly connected, so that the walk will almost surely eventually visit all vertices. The visited set $X$ must have the property that a single walk ending at $u$ could have visited precisely those vertices; in particular, any set $X$ which contains $u$ and induces a strongly connected subgraph is feasible. 
	\begin{labeling}{$\mathtt{NextStep}\left(G,u,X\right)$:}
		\item[$\mathtt{NextStep}\left(G,u,X\right)$:] Output a probability distribution over the neighbours of $u$ (a row of the bias matrix) which minimises the expected time for the $\eps$-TBRW to visit every vertex not in $X$, assuming an optimal strategy is followed thereafter.
	\end{labeling}
	Any such problem may arise during the $\eps$-TBRW on $G$ starting from some vertex in $X$, no matter what strategy was followed up to that point, since with positive probability the bias coin did not allow the controller to influence any previous walk steps. We also introduce the following decision version of $\mathtt{NextStep}\left(G,u,X\right)$ for $X\subset V$, $u\in X$ and $y,z \in \Gamma(u)$: 
	\begin{labeling}{$\mathtt{BestStep}\left(G,X,y,z\right)$:}
		\item[$\mathtt{BestStep}\left(G,X,y,z\right)$:] Is $\tetp(y,X\cup\{ y\})< \tetp(z,X\cup\{ z\}) $? 
	\end{labeling}
We can also consider the decision problem for the expected time to cover a given unvisited set $X$ from a vertex $u$: 
	\begin{labeling}{$\mathtt{Cost}\left(G,u,X,y,z\right)$:}
	\item[$\mathtt{Cost}\left(G,u,X,C\right)$:] Is $\tetp(u,X)< C$? 
\end{labeling}
	We show that all three of the problems above can be solved in polynomially bounded space. 
		\begin{theorem}\label{covinPSPACE}Let $G$ be any strongly connected directed weighted graph and $u \in V$ and $X\subseteq V$ be any connected vertex subset containing $u$. Further let $x,y \in \Gamma(u)$ and  $C<\infty$. Then $\mathtt{Cost}\left(G,u,X,C\right)$ and $\mathtt{BestStep}\left(G,X,x,y\right)$ are in $\PSPACE$.
	\end{theorem}

\begin{remark}\label{rmk:nextorbest}Note that $\mathtt{NextStep}(G,u,X)$ is not a decision problem, and so not in $\PSPACE$; however, it can be solved by using a polynomial number of 
	calls to $\mathtt{BestStep}$ to identify an optimal neighbour of $u$. This is since there is an optimal solution to $\mathtt{NextStep}$ supported on a single neighbour by \cref{Covunchanging}.
\end{remark} 
We show all three problems are $\PSPACE$-hard, thus $\mathtt{Cost}$ and $\mathtt{BestStep}$ are $\PSPACE$-complete.
	
	\begin{theorem}\label{allhard} For any fixed $\eps\in(0,1)$ the problems $\mathtt{Cost}$, $\mathtt{BestStep}$ and $\mathtt{NextStep}$ are $\PSPACE$-hard on strongly connected directed graphs. 
	\end{theorem}
 In \cite{ITCSpaper} we proved that the $\mathtt{NextStep}$ problem for the related CRW on undirected graphs is $\NP$-hard. The same argument holds for the $\eps$-biased random walk and in \cref{S:NPadapt} we shall provide some details of how to adapt the proof to give the following. 
	\begin{theorem}\label{NextIsNPHard} For any fixed $\eps\in(0,1)$ the problems $\mathtt{Cost}$, $\mathtt{BestStep}$ and $\mathtt{NextStep}$ are $\NP$-hard on undirected graphs, even under the restriction $\dmax\leq 3$. \end{theorem}
	 In a similar vein, the proofs of Theorems \ref{allhard} and \ref{covinPSPACE} can also be fairly easily adapted so the same results hold for the CRW of \cite{POTC,ITCSpaper}.
	
	\subsection{Properties of Optimal Covering Strategies}\label{covalgsec}
	The following result from \cite{POTC} says that one can encode the cover time problem as a hitting time problem on a (significantly) larger graph. In \cite{POTC} this is proved for the CRW; the same proof applies to the $\eps$-TBRW.
	\begin{lemma}[Lemma 7.7 of \cite{POTC}]\label{covashit}
		For any graph $G=(V,E)$ let the (directed) auxiliary graph $\tilde{G}=(\tilde{V},\tilde{E})$ be given by $\tilde{V}=V\times \mathcal{P}(V)$ (where $\mathcal{P}(V)$ is the power set) and $\tilde{E}= \left\{((i,S),(j,S\cup j))\mid ij \in E,S\subseteq V\right\}$. Then solutions to $\mathtt{Cov}\left(G,v\right)$ correspond to solutions to $\mathtt{Hit}\bigl(\tilde{G},(v,\{v\}),W\bigr)$ and vice versa, where $W=\{(u,V)\mid u\in V\}$.
	\end{lemma}
	
	Recall that if the next step is a bias step then the $\eps$-TBRW strategy will output a probability distribution over the neighbours of the current vertex which depends on the history of the walk.  
	\begin{corollary}\label{Covunchanging}There exists an optimal strategy for the $\eps$-TBRW cover time problem which is unchanging between times when a new vertex is visited. Moreover, given a fixed visited set $X$, for each vertex $x\in X$ there is fixed $y\in \Gamma(x)$ such that whenever the walk is at $x$ the distribution over neighbours of $x$ given by the strategy is $\delta_y$, that is it always moves to $y$ when given the choice. 
	\end{corollary}
	\begin{proof}[Proof of \cref{Covunchanging}]We shall appeal to \cref{covashit} and consider the problem of covering $G$ as hitting the set $W$ in the auxiliary graph $\tilde{G}$. This is now an instance of the optimal first-passage problem in the context of Markov decision processes \cite{Derman} (see also \cite{ABKLPbias}), and the existence of a time independent deterministic optimal policy follows from \cite[Thm.\  3, Ch.\ 3]{Derman}. 
		
		Regarding time independence, notice that although the strategy for hitting the vertex $W$ in $\tilde{G}$ is independent of time this is not strictly true of the original cover time problem. Recall $\tilde{G}$ is a directed graph which consists of a series of undirected graphs linked by directed edges, the undirected graphs represent the sub-graphs of $G$ induced by possible visited sets and the directed edges correspond to the walk in $G$ visiting a new vertex. Since the strategy for $\tilde{G}$ is independent of time, between the times when a new vertex is added to the covered set the strategy on $G$ is fixed. 
		
		Regarding the term deterministic; using the terminology from \cite{Derman}, the set of actions at a given time are the neighbours of current vertex policy/strategy is a probability distribution over the set of actions. Derman \cite{Derman} states that a policy is deterministic if at every possible step in the process these distributions are supported on  a single action. Since in our case there is a function taking the vertices of $\tilde{G}$ to those of $G$ this corresponds to a strategy always choosing the same fixed neighbour of a given vertex during epochs when the visited set does not change.\end{proof}

	\subsection{The \texttt{BestStep} and \texttt{Cost} problems are in \PSPACE}
	In light of \cref{covashit} we can solve $\mathtt{Cov}(G,v)$ in exponential time using \cref{azarpoly}, by solving the associated hitting time problem on the (exponentially sized) auxiliary graph $\tilde{G}$. We shall now prove that the problems $\mathtt{BestStep}, \mathtt{NextStep}$ and $\mathtt{Cost}$ can be solved using polynomially bounded space for any finite irreducible Markov chain, where that $\mathtt{NextStep}$ equates to computing the optimal strategy for one step in the on-line cover time problem.

	\begin{proof}[Proof of \cref{covinPSPACE}]For a set $S\subset V$ let $\tetp(s,S)  $ be the optimal expected cover time of $G$ from $s\in S$ by the $\eps$-TBRW assuming that $S$ has already been visited. Let $\partial S=  \{y \in V\setminus S: \exists x\in S:xy\in E \}$. By \cref{Covunchanging} if we consider steps of the walk between times when a new vertex is added to the set of visited vertices then the strategy can be just thought of as a fixed bias matrix.
		\begin{claim}Let $S \subset V$, $ s \in S$ and assume for each $x \in \partial S$ we have access to the value $\tetp(x,S\cup\{ x\} )  $. Then we can compute $\tetp(s,S)$ and a bias matrix $\mathbf{B}$, which is a optimal bias matrix while $S$ is the visited set, in $\poly(n)$ space.  \end{claim} 
		\begin{poc} Given $S\subset V$, $s \in S$ and a bias matrix $\mathbf{B}$, let $\tetp(s,S,\mathbf{B}) $ be the expected cover time from $s$ assuming that $S$ has been covered and strategy $\mathbf{B}$ is followed until the first time the walk exits $S$ and an optimal strategy is followed thereafter. If follows that \begin{equation}\label{etep}\tetp(s,S) = \inf_{\mathbf{B}}\tetp(s,S,\mathbf{B}),\end{equation}where the infimum is over stochastic matrices supported on the edges of $G$. Since $G$ is strongly connected the random walk on $G$ is irreducible and so for any $\eps<1$ and $\mathbf{B}$ it follows that  $\tetp(s,S,\mathbf{B})$ is at most polynomial in $n$. 
			
			The idea is that for a fixed $\mathbf{B}$, $\tetp(s,S,\mathbf{B})$ is the solution to a discrete harmonic equation with boundary values $\{\tetp(x,S\cup\{ x\} ) \}_{x \in \partial S }$. 	Indeed, let $\mathbf{P}$ be the transition matrix of the SRW on $G$, and $ h_x :=\tetp(x,S\cup \{x\} ) $ for any $x \in S \cup \partial S $. Then 
			\[h_x=\begin{cases}1+\sum_y \left(p_{xy}(1-\eps) + \eps b_{x,y}\right) \cdot h_y&\quad\text{if }x\in S\\
			\tetp(x,S\cup\{x\} )&\quad\text{if }x\in \partial S.\end{cases}\]
			We can then solve this in polynomial space since the values  $\{\tetp(x,S\cup\{ x\} ) \}_{x \in \partial S }$ are known. Since by \cref{Covunchanging} there is an optimal strategy minimising cover time where the bias distributions are only supported on a single neighbour, it suffices to only consider matrices $\mathbf{B}$ with a single $1$ in each column. There are at most $n^n$ of these and so by \cref{etep} we can determine $\tetp(s,S)$ by calculating $\tetp(s,S,\mathbf{B})$ for each such $\mathbf{B}$ sequentially and only storing the best pair $\mathbf{B}$, $\tetp(s,S,\mathbf{B})$ found so far. \end{poc}
	 
		We now use the claim to show that we can calculate the value $\tetp(u,X)$ in $\poly(n)$ space, consisting of the space required for the claim plus additional space to store up to $n^2$ other values, for each pair $u,X$. To be precise, we prove by induction on $n-\abs{X}$ that we may calculate $\tetp(u,X)$ using additional storage for at most $(n-\abs{X})n$ other values. If $\abs{X}=n$ then $X=V$ and $\tetp(u,X)=0$. If $\abs{X}=n-k$ and the result holds for all larger sets then we may compute each of $\tetp(x,X\cup\{x\})$ for $x\in\partial X$ using only $(k-1)n$ additional storage spaces, storing the results in at most $n$ further storage spaces, and then use the claim to compute $\tetp(u,X)$ from these values. Thus the result holds for all pairs $u,X$ by induction, and so computing $\tetp(u,X)$ and comparing it with $C$ solves $\mathtt{Cost}\left(G,u,X,C\right)$ in $\poly(n)$ space.

		The claim also gives us the matrix $\mathbf{B}$ minimising $\tetp(u,X,\mathbf{B})$, and the column of this matrix corresponding to the vertex $u$ solves $\mathtt{NextStep}\left(G,u,X\right)$. Finally, we can solve the problem  $\mathtt{BestStep}\left(G,X,x,y\right)$ by computing both $\tetp(x,X\cup\{ x\})$ and $\tetp(y,X\cup\{ y\})$ and comparing them. 
	\end{proof}

	\subsection{The \texttt{Cost} problem is \PSPACE-hard}
	
We aim to show that $\mathtt{Cost}$ is $\PSPACE$-hard via a reduction to quantified satisfiability, which is the canonical \PSPACE-complete problem \cite{AroraBarak}. To define this problem let $\phi$ be a conjunctive normal form for variables $x_1,\ldots,x_{2n}$, where we can assume that each clause contains three literals. The decision problem is then as follows.  

	  \begin{labeling}{$\mathtt{QSAT}(\phi)$:}
	\item[$\mathtt{QSAT}(\phi)$:] $\exists x_1, \forall x_2,\exists x_3, \dots ,  \forall x_{2n}$ such that $\phi(x_1, x_2, \dots , x_{2n})$ holds?
\end{labeling}  
Let $N(\phi,x)$ be the number of clauses of $\phi$ featuring the literal $x$ (where $x\in\{x_i,\overline{x_i}\mid i\in\{1,\ldots,n\}\}$) and $C$ be the total number of clauses. We can assume that no two complementary literals $x_i$ and $\overline{x_i}$ appear in the same clause, since otherwise this clause is trivially satisfied. We shall now introduce some gadgets which will help us make the reduction between the two problems. For simplicity, we shall assume $\eps=1/4$ throughout; the proof can be adapted to a general constant value of $\eps$ with suitable changes to the length parameters $\ell$ of the various gadgets.

	\subsubsection{The Gadgets}

\begin{gadget}{The Quincunx Gadget $Q(\ell)$} This gadget allows the walker to choose between two alternatives with very high probability. It consists of vertices $v_{i,j}$ for $0\leq i\leq j\leq\ell$, where the parameter $\ell$ is an odd integer, together with two other vertices $x,y$. The walker enters at $v_{0,0}$ and leaves at either $x$ or $y$. Each vertex $v_{i,j}$ for $j<\ell$ has two outedges to $v_{i,j+1}$ and $v_{i+1,j+1}$; each vertex $v_{i,\ell}$ has a single outedge, which goes to $x$ if $2i<\ell$ and to $y$ if $2i>\ell$. We refer to $v_{0,0}$ as the ``entrance'', $x$ as the ``left exit'' and $y$ as the ``right exit''. Note that the time taken to cross the quincunx is $\ell+1$ deterministically.\end{gadget}
\begin{lemma}\label{quincunx}
If the controller of the $1/4$-TBRW wishes to exit $Q(\ell)$ at $x$ (or $y$) then they may achieve this with probability at least $1-0.99^{\ell}$.
\end{lemma}
\begin{proof}
We think of each step from $v_{i,j}$ to $v_{i,j+1}$ as moving ``left'', and each step from $v_{i,j}$ to $v_{i+1,j+1}$ as moving ``right''. In order to maximise the probability of exiting at $x$, the controller should choose to move left whenever possible. In this case the number of times the walk moves right, $R$, is given by a binomial random variable with mean $\mu=3\ell/4$, and by the multiplicative Chernoff bound (see e.g~\cite[Thm.\ 4.4]{MitzUpfal}) 
\[\Pr{R>\frac{\ell}{2}}=\Pr{R>\mu/3}<\left(e^{1/3}(3/4)^{4/3}\right)^\mu=\left(3e^{1/4}/4\right)^{\ell/2}<0.99^\ell.\qedhere\]
\end{proof}

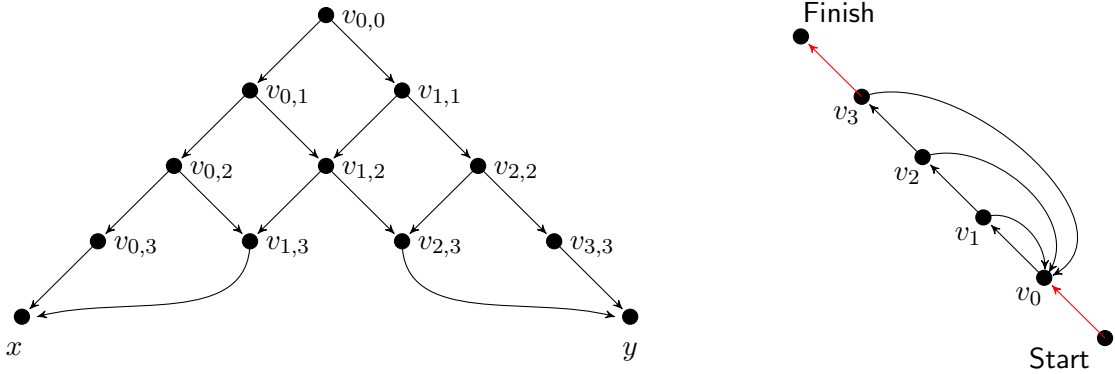
\begin{figure}
	\begin{subfigure}{.65\textwidth}
		
		\begin{tikzpicture}[label/.style={thick,circle}]
		\usetikzlibrary{arrows.meta}
		\usetikzlibrary{decorations.markings}
		\usetikzlibrary{decorations.pathreplacing}
		\tikzset{->-/.style={decoration={
					markings,
					mark=at position .5 with {Stealth[length=4mm]}},postaction={decorate}},>=stealth'}

		\foreach \y in {0,1,2,3}
		\foreach \x in {0,...,\y}{
			\draw[fill] (2*\x-\y,-\y) circle (.1);
			\draw (2*\x-\y+.5,-\y+.2) node[anchor=north]{{$v_{\x,\y}$}};}
		
		\foreach \x in {0,4}{
			\draw[fill] (2*\x-4,-4) circle (.1);
		}
		\draw (-4.1 ,-4.2) node[anchor=north]{{$x$}};
		\draw (4 ,-4.2 ) node[anchor=north]{{$y$}};
		\draw[->] (-3,-3)  -- (-3.9,-3.9);
		\draw[->] (3,-3)  -- (3.9,-3.9);
		\draw[->] (-1,-3) to[out=270,in=15]  (-3.8,-4.);
		\draw[->] (1,-3)  to[out=270,in=165] (3.8,-4);
		
		\foreach \y in {0,1,2}
		\foreach \x in {0,...,\y}{
			\draw[->] (2*\x-\y,-\y)  -- (2*\x-\y-.9,-\y-.9);
			\draw[->] (2*\x-\y,-\y)  -- (2*\x-\y+.9,-\y-.9);}
		\end{tikzpicture}
		
	\end{subfigure}
	\begin{subfigure}{.3\textwidth}

		\begin{tikzpicture}[label/.style={thick,circle}]
		\usetikzlibrary{arrows.meta}
		\usetikzlibrary{decorations.markings}
		\usetikzlibrary{decorations.pathreplacing}
		\tikzset{->-/.style={decoration={
					markings,
					mark=at position .5 with {Stealth[length=4mm]}},postaction={decorate}},>=stealth'}

		\foreach \x in {2,...,7}{
			\draw[fill] (4*\x/5- 3/5 ,-4*\x/5) circle (.1);}
		
		\foreach \x in {0,...,3}{
			\draw (12/5-4*\x/5+1.6 ,4*\x/5-24/5) node[anchor=north]{{$v_{\x}$}};}
		
		\draw (4.4,-5.6) node[anchor=north]{$\mathsf{Start}$};
		\draw (1.5 ,-1) node[anchor=north]{$\mathsf{Finish}$};
		
		\foreach \x in {3,...,5}{
			\draw[->] (4*\x/5+1/5 ,-4*\x/5-4/5) --(4*\x/5- 3/5 +.1,-4*\x/5-.1)  ;}
		
		\draw[red,->] (9/5,-12/5) -- (1.1 ,-8/5-.1);
		\draw[red,->] (5 ,-28/5)--(21/5+.1 ,-24/5-.1);

		\foreach \x in {2,...,4}{
			\draw[->] (4*\x/5+1/5 ,-4*\x/5-4/5) to[out=20,in=\x*30 -30 ]  (4.41-\x/20 ,-4.85+\x/25);
		}

		\end{tikzpicture}

	\end{subfigure}
	\caption{A Quincunx Gadget $Q(3)$ (left) and a Slow Path Gadget $P(3)$ (right). Removing $\mathsf{start}$, $\mathsf{finish}$ and the adjacent red edges from $P(3)$ leaves a Steep Hill $H(3)$.}
\end{figure}

\begin{gadget}{The Steep Hill Gadget $H(\ell)$}This consists of vertices $v_0, \dots, v_\ell$ with directed edges $v_{i-1},v_i$ and $v_i,v_0$ for each $i\in\{1,\ldots,\ell\}$. Note that $H(\ell)$ is strongly connected, but (for $\ell>1$) it is much easier to reach $v_0$ from $v_{\ell}$ than vice versa. We refer to $v_0$ as the ``bottom'' and $v_\ell$ as the ``top''.\end{gadget}
	
\begin{gadget}{The Slow Path Gadget $P(\ell)$}This consists of a steep hill $H(\ell)$ together with two extra vertices, a ``start'' vertex and ``finish'' vertex, and directed edges from the start vertex to the bottom of the hill and from the top of the hill to the finish vertex.\end{gadget}
	
The slow path gadget will play the part of a very long path in the construction which follows; we use a slow path instead of a simple path in order for the (expected) time to traverse to be exponentially large even though the gadget has polynomial size. We calculate the expected time to traverse now.

	\begin{lemma}\label{LengthSlowPath}For any $\eps<1$, the expected time taken for the $\eps$-TBRW to traverse $P(\ell)$ from start to finish, using an optimal strategy, is given by
	\[ L(\ell):=\frac{11}{3}\bfrac{8}{5}^{\ell} - \frac{2}{3}.\]
	\end{lemma}
 \begin{proof}
 Let $H_{i}$ be the expected time for the walk to reach the finish from vertex $v_i$, and set $H_{\ell+1}=0$. Observe that for any $1\leq i\leq\ell $ we have 
  \[ H_{i} =1 + \frac{3}{8}H_0 + \frac{5}{8}H_{i+1}, \] and $H_0 = 1 + H_1$. Using this relation one can show by induction that for any $2\leq j\leq \ell+1$,
  \[ H_0 = 2\bfrac{8}{5}^{j-1} + \sum_{i=1}^{j-2}\bfrac{8}{5}^i + H_{j}.\]
  The result follows by setting $j=l+1$ and summing the geometric series, noting that the expected time to traverse the gadget is $1+H_0$.
  \end{proof} 

\begin{gadget}{The Roundabout Gadget $R(\ell_p,\ell_q,k)$}This consists of a cyclic arrangement of $k$ copies of the slow path $P(\ell_p)$ and $k$ copies of the quincunx $Q(\ell_q)$. Identify the finish vertex of each slow path with the entrance of a quincunx, and identify the right exit of each quincunx with the start vertex of the next slow path. We say that the left exits of the quincunxes are the ``departure vertices'' of the roundabout, and the right exits of the quincunxes are the ``arrival vertices''; arrival and departure vertices are ``corresponding'' if they are exits of the same qunicunx.\end{gadget}

\begin{gadget}{The Star Connector Gadget $S(\ell,k)$}
The purpose of this gadget is to allow us to make the visited set of our graph strongly connected.
It consists of $k$ steep hills $H(\ell)$, with their top vertices identified. The bottoms of the hills we call the ``ports'' of the star connector, and the identified top vertices are the ``nexus''.\end{gadget}
We will use the following simple lemma to bound the time spent inside the star connector.
\begin{lemma}\label{octopus}Consider a star connector $S(\ell,k)$, with each port having at least one outgoing edge to some vertex which is not part of the star connector. Start a $1/4$-TBRW at any port. Then, no matter what strategy is employed, the expected time spent in the star connector before leaving is less than $14$ and the probability of reaching the nexus before leaving is less than $\bfrac{13}{14}^\ell$.
\end{lemma}
\begin{proof}
Note that from any vertex which is not a port, the next step reaches a port with probability at least $\frac 38$, since either there are only two outedges, each chosen with probability at least $\frac38$ and one of which leads to a port, or we are at the nexus and all outedges lead to ports. Similarly, from any port there are two outedges and so the next step leaves the star connector with probability at least $\frac 38$. Consequently, from any vertex in the star connector there is a probability of at least $\bfrac 38^2$ of leaving the star connector within two steps.

It follows that the number of steps taken before leaving is dominated by $2X-1$, where $X$ is a geometric random variable with success probability $\frac{9}{64}$; this has mean $\frac{128}{9}-1<14$. In order to reach the nexus the walk needs to take at least $\ell+1$ steps before leaving, and so the probability of this is bounded by $\Pr{X>\lceil\ell/2\rceil}\leq\bfrac{55}{64}^{\ell/2}<\bfrac{13}{14}^{\ell}$.
\end{proof}

\begin{figure}
	\begin{subfigure}{.6\textwidth}
		\begin{tikzpicture}[label/.style={thick,circle}]
		\usetikzlibrary{arrows.meta}
		\usetikzlibrary{decorations.markings}
		\usetikzlibrary{decorations.pathreplacing}
		\tikzset{->-/.style={decoration={
					markings,
					mark=at position .5 with {Stealth[length=4mm]}},postaction={decorate}},>=stealth'}

		\def \n {27}
		\def \radius {2.5cm}
		\def \margin {4} 
		
		\foreach \s in {1,...,\n}
		{
			\draw[fill] ({360/\n * (\s - 1)}:\radius) circle (.08);
		}

		\foreach \s in {1,...,\n}
		{
			\draw[->] ({360/\n * (\s - 1)}:\radius) 
			arc ({360/\n * (\s - 1)}:{360/\n * (\s)-2.5}:\radius);
		}
		
		\foreach \x in {0,9,18}
		\foreach \s in {6,7,8}
		{
			\draw[->] ({360/\n * (\s +\x)}:\radius)   to[out=180/\n * \s +360/\n *\x +250 ,in=1200/\n * \s +360/\n *\x-80] ({360/\n * \x +380/\n *5 -\s/2}:\radius-2.6) ;
		}

		\foreach \x in {0,9,18}
		{
			
			\draw[fill,green] ({360/\n * (\x+4) }:\radius ) circle (.08);
			\foreach \r in {1,...,4}{
				\def \dista {.6cm}
				\draw[fill] ({360/\n * (\x+\r/1.5) }:\radius + \r*\dista) circle (.08);
				\draw[fill,red] ({360/\n * (\x+8/3) }:\radius + 4*\dista) circle (.08);
				\draw[->] ({360/\n * (\x+\r/1.5-2/3) }:\radius  + \r*\dista -\dista) -- ({360/\n * (\x+\r/1.5)-1.5}:\radius + \r*\dista-1.5);
			}
			
			\def \distb {.6cm}
			
			\draw[->] ({360/\n * (\x+2/3) }:\radius + \distb ) -- ({360/\n * (\x+2/3+1) -1.5}:\radius + \distb );
			
			\draw[->] ({360/\n * (\x+4/3) }:\radius + 2*\distb ) -- ({360/\n * (\x+4/3+1) -1.5}:\radius + 2*\distb );
			
			\foreach \r in {1,2}{
				
				\draw[fill] ({360/\n * (\x+\r/1.5 +1) }:\radius + \r*\distb) circle (.08);
				\draw[->] ({360/\n * (\x+\r/1.5-2/3 + 1) }:\radius  + \r*\distb -\distb) -- ({360/\n * (\x+\r/1.5+1)-1.5}:\radius + \r*\distb-1.5);
			}
			
			\draw[->] ({360/\n * (\x+2/3+1) }:\radius + \distb ) -- ({360/\n * (\x+2/3+2) -1.7}:\radius + \distb );
			
			\draw[->] ({360/\n * (\x+2/1.5 +1) }:\radius + 2*\distb)  to[out=+360/\n *\x +60 ,in=360/\n *\x-130] ({360/\n * \x +960/\n-.5}:\radius + 4*\distb-2.5) ;
			
			\foreach \r in {1}{
				\draw[fill] ({360/\n * (\x+\r/1.5 +2) }:\radius + \r*\distb) circle (.08);
				\draw[->] ({360/\n * (\x+\r/1.5-2/3 + 2) }:\radius  + \r*\distb -\distb) -- ({360/\n * (\x+\r/1.5+2)-1.5}:\radius + \r*\distb-1.5);
				
				\draw[->] ({360/\n * (\x+\r/1.5 +2) }:\radius + \r*\distb)  to[out=+360/\n *\x +120 ,in=360/\n *\x +30] ({360/\n * \x +380/\n *4-4}:\radius+2.7) ;
			}
			
		}

		\end{tikzpicture}%
	\end{subfigure}%
	\begin{subfigure}{.4\textwidth}
		\begin{tikzpicture}[label/.style={thick,circle}]
		\usetikzlibrary{arrows.meta}
		\usetikzlibrary{decorations.markings}
		\usetikzlibrary{decorations.pathreplacing}
		\tikzset{->-/.style={decoration={
					markings,
					mark=at position .5 with {Stealth[length=4mm]}},postaction={decorate}},>=stealth'}

		\foreach \x in {3,...,6}{
			\draw[fill] (4*\x/5- 3/5 ,-4*\x/5) circle (.1);}

		\foreach \x in {3,...,5}{
			\draw[->] (4*\x/5+1/5 ,-4*\x/5-4/5) --(4*\x/5- 3/5 +.1,-4*\x/5-.1)  ;}

		\foreach \x in {2,...,4}{
			\draw[->] (4*\x/5+1/5 ,-4*\x/5-4/5) to[out=-30,in=\x*30 -30 ]  (4.41-\x/20 ,-4.85+\x/25);
		}
		 
		\foreach \x in {3,...,6}{
			\draw[fill] (9/5 ,-\x+3/5) circle (.1);}

		\foreach \x in {3,...,5}{
			\draw[->] (9/5 ,-\x-2/5) --(9/5,-\x+3/5-.15)  ;}
		
		\foreach \x in {2,...,4}{
			\draw[->] (9/5 ,-\x-2/5) to[out=290,in=\x*30-70  ]  (1.95 -\x/50 ,-5.5 +\x/20);}

		\foreach \x in {3,...,6}{
			\draw[fill] (-4*\x/5+ 21/5 ,-4*\x/5) circle (.1);}

		\foreach \x in {3,...,5}{
			\draw[->] (-4*\x/5+17/5 ,-4*\x/5-4/5) --(-4*\x/5+ 21/5 -.1,-4*\x/5-.1)  ;}

		\foreach \x in {2,...,4}{
			\draw[->] (-4*\x/5+17/5 ,-4*\x/5-4/5) to[out=-115,in=-110+\x*30 ]  (-.58+\x/40 ,-4.96+\x/25);
		}
		
		\draw (1,-2) node[anchor=north]{{Nexus}};
		\draw (1.8,-5.7 ) node[anchor=north]{{Ports}};
		\end{tikzpicture}
	\end{subfigure}
	\caption{A Roundabout Gadget $R(3,3,3)$ (left), with arrival vertices in \textcolor{green}{green} and departure vertices in \textcolor{red}{red}, and a Star Connector Gadget $S(3,3)$ (right).}
\end{figure}
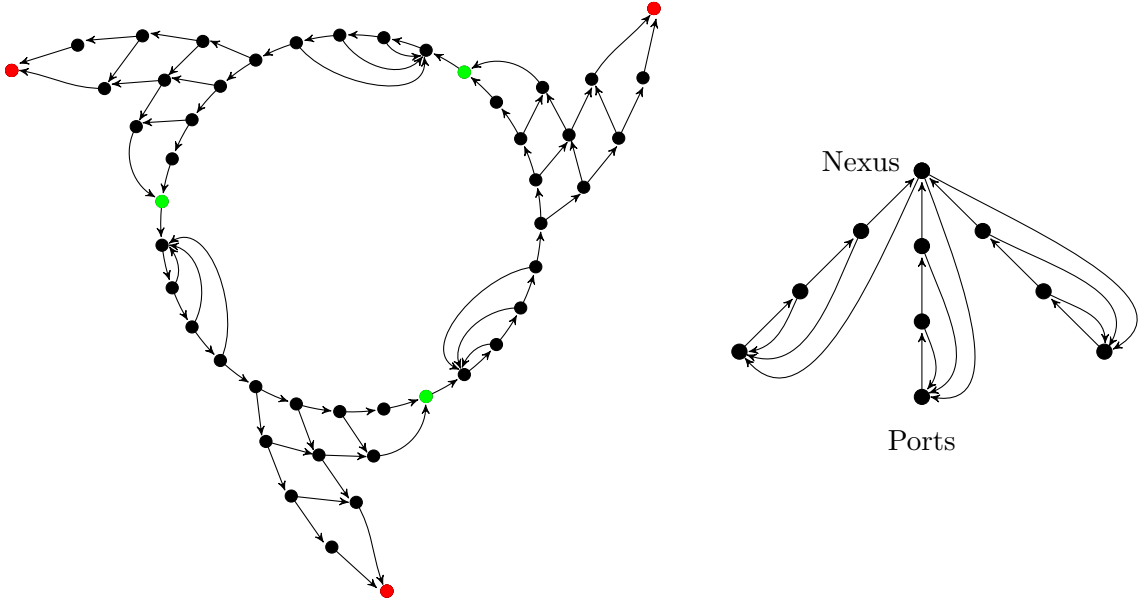

We are now able to describe how we encode an instance of $\mathtt{QSAT}$ as a graph.  

\begin{gadget}{The QSAT Graph $G(\phi)$}
We shall encode a given $\mathtt{QSAT}$ problem on an $n$ variable $3$-CNF $\phi$ with $r$ clauses as the QSAT Graph $G(\phi)$ with a certain unvisited set $X$. We shall build this up in stages. The construction depends on certain length parameters $\ell_p,\ell_q,\ell_s$ for the gadgets which we choose later.  

For each clause take one roundabout gadget $R(\ell_p,\ell_q,3)$ and label its arrival vertices with the literals appearing in that clause. Take a star connector gadget $S(\ell_s,6r)$, and identify its ports with the start vertices of the slow paths and the entrances of the quincunxes in these roundabout gadgets. Mark as unvisited every vertex, other than the start vertices, in the slow paths of the roundabout gadgets. These will form the entire unvisited set $U$.

For each literal $x$, we construct a chain of $N(\phi,x)$ quincunxes $Q(\ell_q)$ as follows. For each clause containing $x$ in turn, take a quincunx and two slow paths. Identify the right exit of the quincunx with the start vertex of one of the slow paths, and identify the end vertex of that slow path with the arrival vertex of the clause roundabout labelled with $x$. Identify the corresponding departure vertex with the start vertex of the other slow path, and identify the end vertex of that slow path with the left exit of the quincunx. Add a directed edge from the left exit of the quincunx to the entrance of the next quincunx; for the final quincunx, instead add a directed edge to a new vertex $\mathsf{out}_x$. Label the entrance of the first qunicunx as $\mathsf{in}_x$. We refer to this chain of qunicunxes as the $x$-cascade.

Now, for each $i\leq 2n$ we connect the $x_i$-cascade and the $\overline{x_i}$-cascade as follows. Identify $\mathsf{out}_{x_i}$ and $\mathsf{out}_{\overline{x_i}}$ to form a new vertex $\mathsf{last}_i$. If $i$ is even, add a new vertex $\mathsf{first}_i$ with directed edges to $\mathsf{in}_{x_i}$ and $\mathsf{in}_{\overline{x_i}}$. If $i$ is odd, instead add a quincunx, with entrance $\mathsf{first}_i$ and left and right exists identified with $\mathsf{in}_{x_i}$ and $\mathsf{in}_{\overline{x_i}}$. The odd values are the existentially quantified variables, and here the controller has a very high probability of being able to choose whether to set $x_i$ as true or false; for even values (universally quantified) this choice is approximately random, and the controller must therefore cope with an unfavourable sequence of choices for these variables with some probability which is not too small.

Finally, for each $i<2n$ identify $\mathsf{last}_i$ and $\mathsf{first}_{i+1}$. Add a slow path from $\mathsf{last}_{2n}$ to $\mathsf{first}_1$. Designate $\mathsf{first}_1$ as the starting vertex of the walk.\end{gadget}

\begin{figure}[ht]
	\begin{tikzpicture}[scale=.83]
	\usetikzlibrary{shapes.geometric}
	\usetikzlibrary{arrows.meta}
	\usetikzlibrary{decorations.markings}
	\usetikzlibrary{decorations.pathreplacing}
	
	\tikzset{->-/.style={decoration={
				markings,
				mark=at position .5 with {Stealth[length=9mm]}},postaction={decorate}},>=stealth'}
	
	\path [use as bounding box] (-2.5,-3.5) rectangle (16.5,12.8);	
	
		\begin{scope}

	
	\draw[thick,red!80,->>] (8.55,10)  to[out=0,in=100] (10.95,7.7);
	\draw[thick,red!80,->>] (10.55,7.85)  to[out=120,in=0] (7.8,9.53);
	
	\draw[thick,red!80,->>] (8.55,8.5)  to[out=340,in=100] (10.95,1.2);
	\draw[thick,red!80,->>] (10.55,1.35)  to[out=105,in=0] (7.8,8.03);
	
	\draw[thick,red!80,->>] (1.5,8.85)  to[out=0,in=220] (4.18,10-.866*.5-.05);		
	\draw[thick,red!80,->>] (3.5,10)  to[out=160,in=70] (1.05,8.7);	
	
	\draw[thick,red!80,->>]  (0.05,6.3) to[out=350,in=200] (4.2,5.03);	
	\draw[thick,red!80,->>] (3.5,5.5) to[out=180,in=350](0.05,6.75)  ;

	\draw[thick,red!80,->>]  (12, 5.3) to[out=200,in=340](7.8,5.5-.866*.5-.05 ) ;	
	\draw[thick,red!80,->>] (8.55,5.5)  to[out=350,in=200](11.95,5.75)  ;

	\draw[thick,red!80,->>] (8.55,1.98) to[out=340,in=200] (11.94,-.75) ;
	\draw[thick,red!80,->>] (11.94,-1.2) to[out=200,in=350]  (7.82,2-.866*.5) ;
	
	
	\draw[thick,red!80,->>] (-1.5,8.8)  .. controls (-2.5,2.5) and (3,1.7) ..(4.2,1.53);
			\draw[thick,red!80,->>] (3.5,2) ..  controls  (-2,4) and  (-1.6,7.8) ..(-1.15,8.6) ;

	
	\draw[thick,red!80,->>] (8.52,-1) .. controls (12,-3) and (14.5,0) ..(13.1,1.1);
	\draw[thick,red!80,->>] (13.5,1.3) .. controls (15,-1)  and (12,-3.5) .. (7.8,-1.45);

	\draw[thick,red!80,->>] (3.5,-1) ..  controls  (5,-4) and  (20,-6) ..(13.15,7.6) ;
	
	\draw[thick,red!80,->>] (13.45,7.85) .. controls (21,-8.5) and (6,-1) ..(4.3,-1.42);

	\ReflectedRoundClause{1.2}{0}{8}{$\overline{x_1}\vee x_2\vee \overline{x_3}$}
	
	\RoundClause{1.2}{12}{7}{$x_1\vee \overline{x_2}\vee x_4$}
	
	\RoundClause{1.2}{12}{.5}{$x_1\vee x_3\vee \overline{x_4}$}

	
	\Quincrux{.5}{6}{11.5}{}
	\draw (5.8,12) node[anchor=east]{$\mathsf{First}_{x_1}$};
	\CascadeQuincruxR{.5}{8}{10}{}
	\CascadeQuincruxR{.5}{8}{8.5}{}
	\CascadeQuincruxL{.5}{4}{10}{}
	\draw[fill] (6,7) circle (.05);
	\draw (5.7,7) node[anchor=east]{$\mathsf{Last}_{x_1}/\mathsf{First}_{x_2}$};
	\draw[thick,->] (6+.866*.5,11.25)  -- (7.70,10+.866*.5+.05);
	\draw[thick,->] (6-.866*.5,11.25)  -- (4.30,10+.866*.5+.05);
	\draw[thick,->] (7.75,10+.866*.5) -- (7.75,8.5+.866*.5+.08);
	\draw[thick,->] (7.75,8.5-.866*.5) -- (6+.05,7+.03);
	\draw[thick,->,rounded corners] (4.25,10-.866*.5) -- (4.25,8.5-.866*.5)  -- (6-.05,7+.03);
	\draw (5.4,9.3) node[anchor=west]{{ \Large $ \exists x_1$}};
	\draw (4.2,10.2) node[anchor=west]{$\mathsf{in}_{\overline{x_1}}$};
	\draw (7.8,10.2) node[anchor=east]{$\mathsf{in}_{x_1}$};

	\CascadeQuincruxR{.5}{8}{5.5}{}
	\CascadeQuincruxL{.5}{4}{5.5}{}
	\draw[thick,->] (6 ,7)  -- (7.70,5.5+.866*.5+.05);
	\draw[thick,->] (6 ,7)  -- (4.30,5.5+.866*.5+.05);
	\draw[thick,->]  (7.75,5.5-.866*.5)--(6+.06,4+.03);
	\draw[thick,->] (4.25,5.5-.866*.5)--(6-.06,4+.03);
	\draw (4.2,5.75) node[anchor=west]{$\mathsf{in}_{x_2}$};
	\draw (7.8,5.75) node[anchor=east]{$\mathsf{in}_{\overline{x_2}}$};
	\draw (5.4,5.5) node[anchor=west]{{ \Large $ \forall x_2$}};
	\draw (5.7,4) node[anchor=east]{$\mathsf{Last}_{x_2}/\mathsf{First}_{x_3}$};

	\Quincrux{.5}{6}{3.5}{}
	\CascadeQuincruxR{.5}{8}{2}{}
	\CascadeQuincruxL{.5}{4}{2}{}
	\draw[fill] (6,.5) circle (.05);
	\draw[thick,->] (6+.866*.5, 3.25)  -- (7.70,2+.866*.5+.05);
	\draw[thick,->] (6-.866*.5,3.25)  -- (4.30,2+.866*.5+.05);
	\draw[thick,->]  (7.75,2-.866*.5)--(6+.06,.5+.03);
	\draw[thick,->] (4.25,2-.866*.5)--(6-.06,.5+.03);
	
	\draw (4.2,2.25) node[anchor=west]{$\mathsf{in}_{\overline{x_3}}$};
	\draw (7.8,2.25) node[anchor=east]{$\mathsf{in}_{x_3}$};
	\draw (5.4,2) node[anchor=west]{{ \Large $ \exists x_3$}};
	\draw (5.7,.5) node[anchor=east]{$\mathsf{Last}_{x_3}/\mathsf{First}_{x_4}$};

	\CascadeQuincruxR{.5}{8}{-1}{}
	\CascadeQuincruxL{.5}{4}{-1}{}
	\draw[thick,->] (6 ,.5)  -- (7.70,-1+.866*.5+.05);
	\draw[thick,->] (6 ,.5)  -- (4.30,-1+.866*.5+.05);
	\draw[thick,->]  (7.75,-1-.866*.5)--(6+.06,-2.5+.03);
	\draw[thick,->] (4.25,-1-.866*.5)--(6-.06,-2.5+.03);
	\draw[fill] (6,-2.5) circle (.05);
	\draw (5.7,-2.5) node[anchor=east]{$\mathsf{Last}_{x_4}$};
	\draw (5.4,-1) node[anchor=west]{{ \Large $ \forall x_4$}};
	\draw (4.2,-.75) node[anchor=west]{$\mathsf{in}_{x_4}$};
	\draw (7.8,-.75) node[anchor=east]{$\mathsf{in}_{\overline{x_4}}$};
	\draw[thick,red!80,->>,rounded corners] (6,-2.5) -- (6,-3)  -- (-2,-3) -- (-2,12.6) -- (6,12.6) -- (6,12.07);

	\begin{scope}[shift={(-.7,1.8)}]
	\draw[line width=1.8pt] (-1,0) -- (3.1,0)  -- (3.1,-4.5) -- (-1,-4.5) -- (-1,0) -- (3.1,0);
	
	\draw (.5,-.35) node[anchor=west]{\underline{\textbf{Key}}};
	\draw (0,-1) node[anchor=west]{Quincunx};
	\Quincrux{.3}{-.5}{-1}{}
	
	\RoundClause{.25}{-.5}{-2}{}
	\draw (0,-2) node[anchor=west]{Roundabout};
	
	\draw[->] (-.8,-3.2)--(-.2,-2.8);
	\draw (0,-3) node[anchor=west]{Directed edge};

	\draw[red!80,->>] (-.8,-4.2)--(-.2,-3.8);
	\draw (0,-4) node[anchor=west]{Slow Path};
	
	\end{scope}
	\end{scope}
	
	\end{tikzpicture}
	\caption{The QSAT Graph for the \texttt{QSAT} problem $\exists x_1, \forall x_2, \exists x_3, \forall x_4:\phi(x_1,x_2,x_3,x_4)$, where $\phi(x_1,x_2,x_3,x_4) = \left(\overline{x_1}\vee x_2\vee \overline{x_3}\right) \wedge \left(x_1\vee \overline{x_2}\vee x_4 \right) \wedge \left(x_1\vee x_3\vee \overline{x_4} \right) $. For clarity we omit the star connector, which has six arms attached to each roundabout.}
\end{figure}
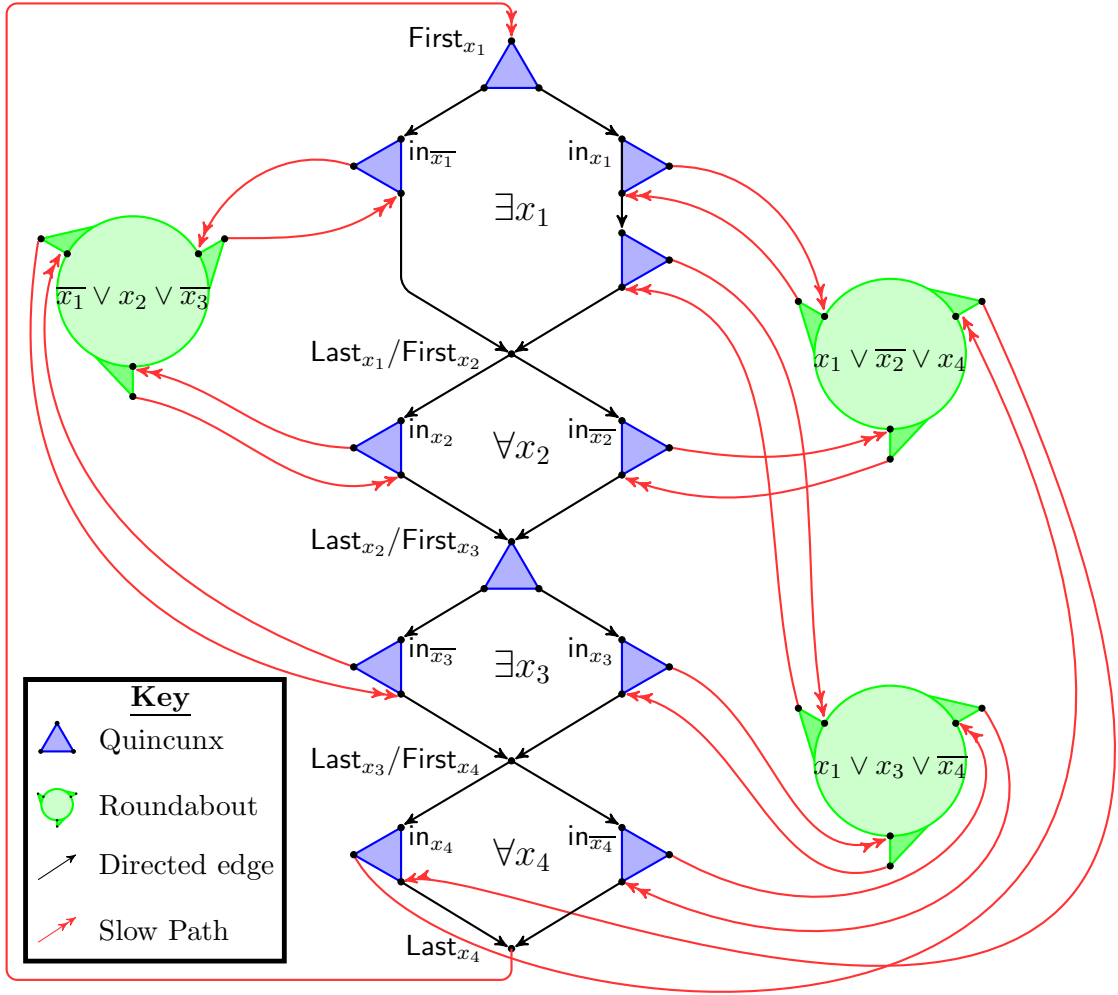

\begin{proof}[Proof of \cref{allhard}]Our analysis of the time taken to cover the unvisited vertices will focus on the number of slow paths traversed (counted with multiplicity). Note that once the walk has crossed the first edge of a slow path, there is no way to leave the whole slow path until it has been entirely traversed, and clearly it is optimal to do so as quickly as possible, taking a random time with expectation $L:=L(\ell_p)$ independently of the decision to start the slow path. 

Suppose that a walker visits the whole set $U$ without visiting the nexus. Then it must have crossed at least $5r-1$ slow paths, since it must cross three slow paths in each roundabout to visit $U$, one slow path to reach each roundabout, and one slow path to leave each roundabout except the last one visited. However, in order to do this crossing exactly $5r-1$ slow paths, the walker visit each roundabout exactly once, and must arrive and depart from each roundabout (except the last) via corresponding vertices, since to do otherwise it would either fail to cross all paths in that roundabout or cross one of them twice. It also cannot cross the slow path from $\mathsf{last}_{2n}$ to $\mathsf{first}_1$. The combination of these factors means that the walker must start from $\mathsf{start}_1$, visit either the $x_1$-cascade or the $\overline{x_1}$-cascade, visit zero or more roundabouts accessible from that cascade, returning to the same cascade each time, then reach $\mathsf{start}_2$ and continue in a like manner, visiting every roundabout before reaching $\mathsf{final}_{2n}$. In particular, the cascades visited correspond to a (possibly incomplete) truth assignment to the variables, and the fact that every roundabout is accessible from some visited cascade means this truth assignment satisfies $\phi$.

The comments above apply to \textit{any} walker; we now analyse the performance of the $1/4$-TBRW. If the instance of $\mathtt{QSAT}$ is satisfiable, then there exists a strategy to visit $U$ while only crossing $5r-1$ slow paths, which succeeds provided the walker avoids the nexus and makes the desired choice from each quincunx encountered. This is because the walker can choose which of the two cascades to visit for each existentially quantified variable, based on which earlier cascades have been visited, in such a way that these cascades give a satisfying assignment, and visit each roundabout at the first opportunity.

We first introduce two ``failure'' events. The first, $F_n$, is that the walker reaches the nexus before crossing $5r$ slow paths. Note that the walker can only enter the star connector at most $10r$ times before crossing $5r$ slow paths, and so \cref{octopus} implies that $\Pr{F_n}<10r\bfrac{13}{14}^{\ell_s}$; this bound is independent of both the strategy followed and the start vertex, provided that this start vertex is outside the star connector or is one of its ports. Setting $\ell_c=a(n+r)$, for some suitable constant $a$, this is less than $\frac{1}{2000r}\bfrac{3}{8}^n$.  

The second failure event, $F_q$, is that the walker fails to make the desired decision at a quincunx on the first occasion that quincunx is traversed. Since there are $6r+n$ quincunxes in the graph in total, this has probability at most $(6r+n)\bfrac{99}{100}^{\ell_q}$ by \cref{quincunx}. Setting $\ell_q=b(n+r)$, for some suitable constant $b$, this is less than $\frac{1}{2000r}\bfrac{3}{8}^n$.

We now bound the expected time for an optimal strategy given that the instance is satisfiable. The walker can succeed in visiting $U$ while crossing exactly $5r-1$ slow paths with probability at least $1-\frac{1}{1000r}\bfrac{3}{8}^n$. We can control the extra time not spent in slow paths while attempting to do this. The walker enters the star connector at most $10r$ times, and each time spends a random amount of time in the star connector. By \cref{octopus}, the expectation of this time is less than $14$. The time spent in quincunxes is at most $(6r+n)\ell_q$, and there are a small number of other steps, at most $3r+2n$, coming from single edges linking quincunxes etc. Thus the expected time for the attempt is at most $(5r-1)L+(n+6r)(\ell_q+30)$. 

If the attempt was unsuccessful, he attempts to ``reset'' by returning to $\mathsf{start}_1$ and restarting. By taking at most one more step, he is outside the star connector or at one of its ports. From here, he can reach $\mathsf{start}_1$ crossing at most three slow paths with probability $1-\frac{1}{1000r}\bfrac{3}{8}^n$. A similar analysis applies to this attempt. Consequently the expected number of attempts taken to return to $\mathsf{start}_1$ is at most $(1-\frac{1}{1000r}\bfrac{3}{8}^n)^{-1}<1.001$, each taking expected time $3L+(n+6r)(\ell_q+30)$. Overall the expected number of additional attempts needed, given that the first failed, is less than $0.001$, and the expected time to ``reset'' after each attempt is less than $1.001(3L+(n+6r)(\ell_q+30))$, giving a total expected time until $U$ is visited of at most \[(5r-1)L+(n+6r)(\ell_q+30)+\frac{1}{1000r}\bfrac{3}{8}^n((5r+5)L+3(n+6r)(\ell_q+30)).\]

We may choose an appropriate constant $c$ and set $\ell_p=c(n+r)$ to satisfy $(n+6r)(\ell_q+30)<\frac{1}{1000}\bfrac{3}{8}^nL$. 
This ensures the value above is at most
\[T_{\mathrm{sat}}:=\left(5r-1+\frac{1}{100}\bfrac{3}{8}^n\right)L.\]

Next we consider the case where the instance of $\mathtt{QSAT}$ is not satisfiable. In that case, no matter how the existentially quantified variables are assigned, there is a way to choose values for the universally quantified variables, depending on values of earlier variables, which avoids $\phi$ being satisfied. As the walker proceeds through the graph, assuming it does not reach the nexus, each universally quantified variable is determined by a single step, and though the controller can influence this step he cannot decrease the probability of either alternative below $\frac38$. Thus, with probability at least $\bfrac38^n$, the truth assignment corresponding to cascades visited does not satisfy $\phi$; recall that in this case the walker must cross at least $5r$ slow paths (or visit the nexus before crossing this number of slow paths, which has probability $\Pr{F_n}$). Thus for the unsatisfiable case the expected time taken is at least 
\[T_{\mathrm{unsat}}:=\left(5r-1+\frac{99}{100}\bfrac{3}{8}^n\right)L.\]

Thus, for these values of $\ell_p,\ell_q,\ell_s$, we have a Cook reduction from $\mathtt{QSAT}(\phi)$ to $\mathtt{Cost}(G(\phi),\allowbreak\mathsf{start}_1,U,(T_{\mathrm{sat}}+T_{\mathrm{unsat}})/2)$, so $\mathtt{Cost}$ is $\PSPACE$-complete.

We next briefly describe how to adapt this argument to prove that $\mathtt{BestStep}$ is $\PSPACE$-hard. Choose a value $\ell'=\BO{n+r}$ to satisfy 
\[\frac13\bfrac{3}{8}^nL<L(\ell')<\frac23\bfrac{3}{8}^nL;\]
this is possible since incrementing $\ell'$ increases $L(\ell')$ by a factor of less than $2$ (and since $\ell'<\ell_p=\BO{n+r}$). We write $L':=L(\ell')$.

Now we modify the construction above to create a graph $G'(\phi)$ as follows. Make each roundabout a copy of $R(\ell_p,\ell_q,4)$ instead of $R(\ell_p,\ell_q,3)$. Add an extra cascade, with extremal vertices labelled $\mathsf{in}_*$ and $\mathsf{out}_*$ connected by slow paths $P(\ell_p)$ to the spare arrival and departure points of every roundabout. Add a new vertex $\mathsf{start}_0$, with two outedges: one to $\mathsf{start}_1$ and the other leading to a slow path $P(\ell')$ which in turn leads to $\mathsf{in}_*$. Finally, add an edge from $\mathsf{out}_*$ to $\mathsf{last}_{2n}$.

In this modified graph, if the walker starts at $\mathsf{start}_1$ the same analysis as above applies, with $(5r-1)L$ replaced by $(6r-1)L$ (to account for the extra slow path in each roundabout). Thus if the instance is satisfiable the expected time started from this point is at most $T_{\mathrm{sat}}+rL$, and if it is not satisfiable it is at least $T_{\mathrm{unsat}}+rL$ (since in order to make use of the new cascade from this starting point, the walker must traverse more than $6r-1$ slow paths). 
However, starting from the beginning of the slow path of length $\ell'$, the expected time is at most $T_{\mathrm{sat}}+rL+L'$ (since after traversing this path the walker can, assuming $F_n$ and $F_q$ do not occur, visit all of $U$ using $6r-1$ other slow paths). It is also at least $(6r-1)L+L'-\Pr{F_n}$. By choice of $L'$ these values lie between $T_{\mathrm{sat}}+rL$ and $T_{\mathrm{unsat}}+rL$.

Thus, starting at $\mathsf{start}_0$, the optimal strategy is to prefer $\mathsf{start}_1$ if the instance is satisfiable and the the other outneighbour if not. This gives a Cook reduction from $\mathtt{QSAT}(\phi)$ to $\mathtt{BestStep}(G'(\phi),\mathsf{start}_0,U)$. Notice that the unique solution to $\mathtt{BestStep}(G'(\phi),\mathsf{start}_0,U)$ is to give full weight to one of the two neighbours, thus both problems are $\PSPACE$-hard. $\PSPACE$-hardness for $\mathtt{BestStep}$  follows from \cref{rmk:nextorbest}.
\end{proof}

	\section{Concluding Remarks and Open Problems}\label{Conclude}
	In this paper we extended the previous work on the $\eps$-biased random walk to include strategies which may depend on the history of the walk. Our motivation for this is the cover time problem for which we obtained bounds using a new technique that allows us relate the probability of any event for the $\epsilon$-biased walk to the corresponding event for a simple random walk. This technique also allowed us to make progress on a conjecture of Azar et al.\ \cite{ABKLPbias}. We note that this conjecture requires some further technical conditions not given in the original statement. However, as discussed in \cref{AzarConjSec}, the only case necessitating this extra condition appears to be that of graphs with large entries in the stationary vector, and we believe that the following slightly refined version of their conjecture should hold.   
	
	\newtheorem*{conj:reformulated}{\cref{reformulated}}
	\begin{conj:reformulated}
		\textit{In any graph a controller can increase the stationary probability of any vertex from $p$ to $p^{1-\eps+\delta} $, where $\delta:=\delta(G)\rightarrow 0$ as $p\rightarrow 0$.}
	\end{conj:reformulated}
	
	We also showed that computing an optimal next step for the $\eps$-TBRW to take in the online version of the covering problem is $\PSPACE$-complete on directed graphs. The class $\PSPACE$ is a natural candidate for the covering problem given that some suitably intricate Markov decision problems and route planning problems are $\PSPACE$-complete \cite{Markovhard}. We believe that the problem is also $\PSPACE$-hard for undirected graphs, although we can only show it is $\NP$-hard. 
	\begin{conjecture}\label{hardconj}
		For undirected graphs $\mathtt{BestStep}$ is $\PSPACE$-hard.
	\end{conjecture}
	
The difficulty in establishing \cref{hardconj} is that on undirected graphs it is difficult to force the walk to make irreversible decisions and so it is not clear how to create gadgets with the sort of one-way nature typical in $\PSPACE$ reductions \cite{DemHenLyn}. In particular there does not seem to be an easy way to adapt our proof for directed graphs to the undirected case. 

	\section*{Acknowledgements}
	J.H.\ was supported by ERC Starting Grant no.\ 639046 (RGGC) and by the UK Research and Innovation Future Leaders Fellowship MR/S016325/1. T.S.\ and J.S.\ were supported by ERC Starting Grant no.\ 679660 (DYNAMIC MARCH). J.S.\ was also supported by EPSRC project EP/T004878/1. J.S.\ would like to thank Dylan Hendrickson and Jayson Lynch for some interesting discussion about $\PSPACE$. We thank Sam Olesker-Taylor for spotting an error in an earlier version of this work.
	\bibliographystyle{abbrv}

\begin{thebibliography}{10}
	
	\bibitem{aldousfill}
	D.~Aldous and J.~A. Fill.
	\newblock Reversible {M}arkov chains and random walks on graphs, 2002.
	\newblock Unfinished monograph, recompiled 2014.
	
	\bibitem{UniTrans}
	R.~Aleliunas, R.~M. Karp, R.~J. Lipton, L.~Lov{\'{a}}sz, and C.~Rackoff.
	\newblock Random walks, universal traversal sequences, and the complexity of
	maze problems.
	\newblock In {\em 20th Annual Symposium on Foundations of Computer Science, San
		Juan, Puerto Rico, 29-31 October 1979}, pages 218--223. {IEEE} Computer
	Society, 1979.
	
	\bibitem{AlonBiasedCoin}
	N.~Alon and M.~O. Rabin.
	\newblock Biased coins and randomized algorithms.
	\newblock {\em Advances in Computing Research}, 5:499--507, 1989.
	
	\bibitem{AroraBarak}
	S.~Arora and B.~Barak.
	\newblock {\em Computational Complexity - {A} Modern Approach}.
	\newblock Cambridge University Press, 2009.
	
	\bibitem{ABKLPbias}
	Y.~Azar, A.~Z. Broder, A.~R. Karlin, N.~Linial, and S.~Phillips.
	\newblock Biased random walks.
	\newblock {\em Combinatorica}, 16(1):1--18, 1996.
	
	\bibitem{ben1987collective}
	M.~Ben-Or and N.~Linial.
	\newblock Collective coin flipping.
	\newblock In S.~Micali, editor, {\em Randomness and Computation}, pages
	91--115. Academic Press, New York, 1989.
	
	\bibitem{Navigate}
	L.~Boczkowski, U.~Feige, A.~Korman, and Y.~Rodeh.
	\newblock Navigating in trees with permanently noisy advice.
	\newblock {\em {ACM} Trans. Algorithms}, 17(2):15:1--15:27, 2021.
	
	\bibitem{BopNarCoin}
	R.~B. Boppana and B.~O. Narayanan.
	\newblock The biased coin problem.
	\newblock In {\em Proceedings of the Twenty-fifth Annual ACM Symposium on
		Theory of Computing}, STOC '93, pages 252--257, New York, NY, USA, 1993. ACM.
	
	\bibitem{CojaEigen}
	A.~Coja{-}Oghlan.
	\newblock On the {L}aplacian eigenvalues of {$G(n, p)$}.
	\newblock {\em Comb. Probab. Comput.}, 16(6):923--946, 2007.
	
	\bibitem{DemHenLyn}
	E.~D. Demaine, D.~H. Hendrickson, and J.~Lynch.
	\newblock Toward a general complexity theory of motion planning: Characterizing
	which gadgets make games hard.
	\newblock In {\em 11th Innovations in Theoretical Computer Science Conference,
		{ITCS} 2020, January 12-14, 2020, Seattle, Washington, {USA}}, pages
	62:1--62:42, 2020.
	
	\bibitem{Derman}
	C.~Derman.
	\newblock {\em Finite State Markovian Decision Processes}.
	\newblock Academic Press, Inc., Orlando, FL, USA, 1970.
	
	\bibitem{Binsearh}
	E.~Emamjomeh{-}Zadeh, D.~Kempe, and V.~Singhal.
	\newblock Deterministic and probabilistic binary search in graphs.
	\newblock In D.~Wichs and Y.~Mansour, editors, {\em Proceedings of the 48th
		Annual {ACM} {SIGACT} Symposium on Theory of Computing, {STOC} 2016,
		Cambridge, MA, USA, June 18-21, 2016}, pages 519--532. {ACM}, 2016.
	
	\bibitem{Feller}
	W.~Feller.
	\newblock {\em An introduction to probability theory and its applications.
		{V}ol. {I}}.
	\newblock John Wiley \& Sons, Inc., New York-London-Sydney, third edition,
	1968.
	
	\bibitem{antblazed}
	E.~Fonio, Y.~Heyman, L.~Boczkowski, A.~Gelblum, A.~Kosowski, A.~Korman, and
	O.~Feinerman.
	\newblock A locally-blazed ant trail achieves efficient collective navigation
	despite limited information.
	\newblock {\em eLife}, 5:e20185, Nov 2016.
	
	\bibitem{GJS-NPC}
	M.~R. Garey, D.~S. Johnson, and L.~Stockmeyer.
	\newblock Some simplified {NP}-complete graph problems.
	\newblock {\em Theoret. Comput. Sci.}, 1(3):237--267, 1976.
	
	\bibitem{ITCSpaper}
	A.~Georgakopoulos, J.~Haslegrave, T.~Sauerwald, and J.~Sylvester.
	\newblock Choice and bias in random walks.
	\newblock In {\em 11th Innovations in Theoretical Computer Science Conference,
		{ITCS} 2020, January 12-14, 2020, Seattle, Washington, {USA}}, pages
	76:1--76:19, 2020.
	
	\bibitem{POTC}
	A.~Georgakopoulos, J.~Haslegrave, T.~Sauerwald, and J.~Sylvester.
	\newblock The power of two choices for random walks.
	\newblock {\em Combinatorics, Probability and Computing, to appear.}, 2021.
	
	\bibitem{GoldPseudo}
	O.~Goldreich.
	\newblock {\em A primer on pseudorandom generators}, volume~55 of {\em
		University Lecture Series}.
	\newblock American Mathematical Society, Providence, RI, 2010.
	
	\bibitem{LocatingTarget}
	N.~Hanusse, D.~Ilcinkas, A.~Kosowski, and N.~Nisse.
	\newblock Locating a target with an agent guided by unreliable local advice:
	how to beat the random walk when you have a clock?
	\newblock In A.~W. Richa and R.~Guerraoui, editors, {\em Proceedings of the
		29th Annual {ACM} Symposium on Principles of Distributed Computing, {PODC}
		2010, Zurich, Switzerland, July 25-28, 2010}, pages 355--364. {ACM}, 2010.
	
	\bibitem{SearchNoise}
	N.~Hanusse, D.~J. Kavvadias, E.~Kranakis, and D.~Krizanc.
	\newblock Memoryless search algorithms in a network with faulty advice.
	\newblock {\em Theor. Comput. Sci.}, 402(2-3):190--198, 2008.
	
	\bibitem{Karp91}
	R.~M. Karp.
	\newblock An introduction to randomized algorithms.
	\newblock {\em Discrete Applied Mathematics}, 34(1-3):165--201, 1991.
	
	\bibitem{BirthdayParadoxPollard}
	J.~H. Kim, R.~Montenegro, Y.~Peres, and P.~Tetali.
	\newblock A birthday paradox for {M}arkov chains with an optimal bound for
	collision in the {P}ollard rho algorithm for discrete logarithm.
	\newblock {\em Ann. Appl. Probab.}, 20(2):495--521, 2010.
	
	\bibitem{levin2009markov}
	D.~A. Levin, Y.~Peres, and E.~L. Wilmer.
	\newblock {\em Markov chains and mixing times}.
	\newblock American Mathematical Society, Providence, RI, 2009.
	\newblock With a chapter by James G. Propp and David B. Wilson.
	
	\bibitem{MitzUpfal}
	M.~Mitzenmacher and E.~Upfal.
	\newblock {\em Probability and Computing: Randomized Algorithms and
		Probabilistic Analysis}.
	\newblock Cambridge University Press, 2005.
	
	\bibitem{Kangaroo}
	R.~Montenegro and P.~Tetali.
	\newblock How long does it take to catch a wild kangaroo?
	\newblock In {\em S{TOC}'09---{P}roceedings of the 2009 {ACM} {I}nternational
		{S}ymposium on {T}heory of {C}omputing}, pages 553--559. ACM, New York, 2009.
	
	\bibitem{MotRag}
	R.~Motwani and P.~Raghavan.
	\newblock {\em Randomized Algorithms}.
	\newblock Cambridge University Press, 1995.
	
	\bibitem{oliveira2018random}
	R.~I. Oliveira and Y.~Peres.
	\newblock Random walks on graphs: new bounds on hitting, meeting, coalescing
	and returning.
	\newblock In {\em 2019 {P}roceedings of the {S}ixteenth {W}orkshop on
		{A}nalytic {A}lgorithmics and {C}ombinatorics ({ANALCO})}, pages 119--126.
	SIAM, Philadelphia, PA, 2019.
	
	\bibitem{Markovhard}
	C.~H. Papadimitriou and J.~N. Tsitsiklis.
	\newblock The complexity of {M}arkov decision processes.
	\newblock {\em Math. Oper. Res.}, 12(3):441--450, 1987.
	
	\bibitem{Pollard}
	J.~M. Pollard.
	\newblock Monte {C}arlo methods for index computation {$\pmod{p}$}.
	\newblock {\em Math. Comp.}, 32(143):918--924, 1978.
	
	\bibitem{SoSt}
	R.~Solovay and V.~Strassen.
	\newblock A fast {M}onte-{C}arlo test for primality.
	\newblock {\em {SIAM} J. Comput.}, 6(1):84--85, 1977.
	
	\bibitem{VazRandPoly}
	U.~V. Vazirani and V.~V. Vazirani.
	\newblock Random polynomial time is equal to slightly-random polynomial time.
	\newblock In {\em 26th Annual Symposium on Foundations of Computer Science,
		Portland, Oregon, USA, 21-23 October 1985}, pages 417--428. {IEEE} Computer
	Society, 1985.
	
\end{thebibliography}

\appendix
\section{Deducing Theorem \ref{trelbdd} from Theorem \ref{nonregboostnew}}\label{S:deduce}
Recall ${p}_{x,\cdot }^{(t)}$ is the distribution of the SRW after $t$ steps started at $x$, and write $\pi(S) = \sum_{s\in S}\pi(s) $ for the stationary probability of a set $S\subseteq V$. We first need a lemma allowing us to approximate instantaneous probabilities by stationary probabilities.
\begin{lemma}[{\cite[Lemma 6.5]{POTC}}]\label{lazyconv} For any graph $G$, $S\subset V $ and $ x \in V$ there exists $t\leq 4\trel\ln n$ such that \[p_{x,S}^{(t)} \geq \pi(S)/3 .\]  
\end{lemma}
Our strategy to bound the cover time will be to emulate the SRW until most of the vertices are covered, only using the additional strength of the $\eps$-TBRW when there are few uncovered vertices remaining. We use the following bound on the duration of the first stage.
\begin{lemma}[{\cite[Lemma 6.6]{POTC}}]\label{vacant} 
	Let $U(t)$ be the number of unvisited vertices at time $t$ by a SRW on a graph and let $T_{n/2^x} $ be the number of SRW steps taken before $U\leq n/2^x  $. Then 
	\[ \Ex{U(2x\cdot \thit)} \leq  \frac{n}{2^x} \qquad \text{and} \qquad \Ex{T_{n/2^x}}\leq 4(x+1)\thit.  \]    
\end{lemma}
\begin{proof}[Proof of Theorem \ref{trelbdd}]
	We first emulate the SRW (i.e.\ at each offered choice, choose independently a uniformly random neighbour) until all but $m=\left\lfloor n/\log^C n\right\rfloor$ 
	vertices have been visited, for some $C$ to be specified later. Let $\tau_1$ be the expected time to complete this phase. Then, by Lemma \ref{vacant}, we have  $\tau_1 \leq 4\thit\cdot C\log_2\log n$. 
	
	We cover the remaining vertices in $m$ different phases, labelled $m,m-1,\ldots,1$, each of which reduces the number of uncovered vertices by $1$. In phase $i$, a set of $i$ vertices are still uncovered, and we write $S_i$ for this set. By Lemma \ref{lazyconv} for any vertex $x$ there is some $t\leq 4\trel\log n $ such that \[p_{x,S_i}^{(t)} \geq \frac{\pi(S_i)}{3} = \frac{1}{3}\cdot \frac{\sum_{s\in S_i} d(s)}{n\davg}\geq\frac{ d_{\mathsf{min}}\cdot i}{3n\davg},\] and thus 
	$q_{u,S_i}^{(t)} \geq \left(\dmin\cdot i/(3n\davg ) \right)^{1-\eps}$ by Theorem \ref{nonregboostnew}. Since from any starting point we can achieve this probability of hitting a vertex in $S_i$ within the next $4\trel\log n$ steps, the expected number of attempts needed to achieve this is at most $\left(\dmin\cdot i/(3n\davg)\right)^{\eps-1}$, meaning that the expected time required to complete phase $i$ is at most
	\[\BO{ \left(\frac{n\cdot \davg}{i\cdot \dmin}\right)^{1-\eps}\cdot \trel\cdot \log n}.\]
	Hence the expected time $\tau_2$ to complete all $m$ phases satisfies
	\begin{align*}
	\tau_2&= \sum_{i=1}^{n/\log^C n} \BO{ \left(\frac{n \davg}{i \dmin}\right)^{1-\eps}\trel\log n}\\
	&= \BO{ \left(\frac{n \davg}{d_{\mathsf{min}}}\right)^{1-\eps} \trel\log n}\sum_{i=1}^{n/\log^C n}i^{\eps-1}. \end{align*}
	Then, since $\sum_{i=1}^{n/\log^C n}i^{\eps-1}\leq \left(n/\log^C n\right)^{\eps}\cdot  \sum_{i=1}^{n/\log^C n}i^{-1} \leq    \left(n/\log^C n\right)^{\eps} \cdot \log n  $, we have
	\begin{align}
	\tau_2&=\BO{\left(\frac{n\davg}{\dmin}\right)^{1-\eps} \trel\log n} \cdot \BO{\left(\frac{n}{\log^C n}\right)^{\eps} \cdot \log n  }\nonumber \\
	&= \BO{   n \cdot \left(\frac{\davg}{\dmin}\right)^{1-\eps}\cdot \trel \cdot \frac{\log^2 n}{\log^{C\cdot \eps} n} }.\label{t2bdd}
	\end{align}
	For the first bound we choose $C=  \log\left( (\frac{\davg}{\dmin})\cdot \trel \cdot \log^2 n \right)/ \left(\eps\cdot \log\log n\right)$ then since we have $\log^{C\cdot \eps}n = (\davg/\dmin)\trel\cdot \log^2 n $ and $\eps>0$ this gives $\tau_2  =\BO{n}$ by \eqref{t2bdd} above. 
	Since in any graph $\thit=\Omega(n) $, the total time is therefore $\BO{\tau_1}$, and for this value of $C$ we have 
	\[\tau_1 = \BO{\frac{\log\left( (\frac{\davg}{\dmin})\cdot \trel \cdot \log^2 n \right)}{\eps\cdot \log\log n} \thit\log\log n} = \BO{\frac{\thit}{\eps}\cdot \log\left( \frac{\davg\cdot \trel \cdot \log  n}{\dmin} \right) } ;\]
as claimed. \end{proof}

\section{\texttt{NextStep} is \textsf{NP}-hard on undirected graphs}\label{S:NPadapt}

\begin{proof}[Proof of \cref{NextIsNPHard}]
	We give a (Cook) reduction from the problem \texttt{Hamilton Path} of deciding if a given graph has a Hamilton path. By \cite{GJS-NPC}, \texttt{Hamilton Path} is $\NP$-complete  even if $H$ is restricted to $\dmax\leq 3$. It suffices to prove the reduction for $\mathtt{BestStep}$ and $\mathtt{Cost}$ by \cref{rmk:nextorbest}. 
	
	Given an instance of \texttt{Hamilton Path}, which is an $n$-vertex graph $H$, we construct the graph $G$ as follows. First replace each edge of $H$ by a path of length $2cn^2$ through new vertices, where $c$ is a suitably large integer to be chosen later. Next add a new pendant path of length $cn^3$ starting at the midpoint of each path corresponding to an edge of $H$. Finally, add edges to form a cycle consisting of the end vertices of these pendant paths (in any order). The construction added in the last two steps is analogous to the star connector as it makes the graph induced by $V(G)/V(H)$ connected. Note that if $H$ has maximum degree $3$, so does~$G$. We now bound the time to cross the paths added in step one.

	\begin{clm}\label{clm1}Let $x,y\in V(H)\subseteq V(G)$ be such that  $xy\in E(G)$. Then, in the graph $G(H)$, we have  \[\frac{2cn^2}{\eps } - \frac{n \sqrt{c}}{\eps^{3/2} }\leq \Heb{x}{V(H)\backslash \{x\}} \leq \Heb{x}{y}\leq \frac{2cn^2}{\eps} + \frac{7}{\eps^2}.\]   
	\end{clm}
	\begin{poc}Let the walk be at a vertex $x\in V(H)$ and the controller's aim be to reach $V(H)\backslash\{x\}$ as quickly as possible. No optimal strategy can use the cycle/star connector as its paths are of length $cn^3$, thus for a lower bound we can remove it. Now, each departure from $x$ puts us on a path of length $2cn^2$ towards a vertex in $V(H)\backslash\{x\}$. Thus, we can consider the problem as a biased walk on the line from $0$ to $2cn^2$. Observe that we can couple an $\eps$-BRW on $\mathbb{Z}$ to an SRW trajectory on $\mathbb{Z}$ by just adding in the bias steps. Thus, if $S_1(k)$ is the earliest time $k$ bias steps have occurred and $S_2(k)$ is first time the SRW from $0$ is at position $k$, then the time for the $\eps$-BRW to hit $y$ stochastically dominates $\min_{0\leq k\leq 2cn^2}\max\{S_1(n-k), S_2(k)\}$. Observe that $\Ex{S_1(k)}=k/\eps$, as this is the expectation of a negative binomial random variable, and $\Ex{S_1(k)} = k^2$ by classical results for SRW on $\mathbb{Z}$ \cite[Chapter XIV]{Feller}. Thus we have  
		\[\Heb xy  \geq\min_{0\leq k\leq 2cn^2}\max\left\{\Ex{S_1(n-k)},\Ex{ S_2(k)}\right\}\geq\min_{0\leq k\leq 2cn^2}\max\left\{\frac{n-k}{\eps}, k^2 \right\}\geq \frac{cn^2}{\eps } - \frac{n \sqrt{c}}{\eps^{3/2} }.   \]
		
		For the upper bound, let $z$ be the vertex halfway along the path from $x$ to $y$, and let $w$ be the neighbour of $x$ on that path. Observe that the time taken to hit $z$ from $x$ is stochastically dominated by $T_1 +T_2$, where $T_1$ the number of steps taken until the last visit to $x$ before hitting $z$ and $T_2$ is the time taken by the $\eps$-BRW on $\mathbb{Z}$ to hit $cn^2$ from $0$. With probability $\eps + (1-\eps)/d(x) $ the walk takes the correct edge towards $y$, thus it leaves $x$ at most $d(x)/(d(x)\eps + 1-\eps) \leq 1/\eps$ times before reaching $w$. The expected time to return to $x$ if the walk does not take the edge $xw$ is bounded by the reciprocal of the stationary probability of $x$ in the $\eps$-BRW walk on $G$ (aiming to reach $x$) with $xw$ removed. This is at most \[ \left(\sum_{i=0}^{cn^2-1} \left(\frac{1-\eps}{1+\eps}\right)^i  + (2cn^2 + n^3)(|E(H)|-1)\cdot\left(\frac{1-\eps}{1+\eps}\right)^{cn^2}\right) \leq \frac{1+\eps}{\eps} + 3cn^5e^{-\eps cn^2},    \] which is less than $3/\eps$ for suitably large $n$. Thus $\Ex{T_1}\leq 3/\eps^2$. Now, by results for the gambler's ruin problem from \cite[Chapter XVI.3]{Feller} we have $\Ex{T_2}\leq cn^2/\eps$. It follows that $\Heb xz\leq cn^2/\eps + 3\eps^2$. Once the walk is at $z$ the probability that it hits $x$ before hitting $y$ is at most $e^{-C\eps c n^2 }$ for some constant $C>0$. Thus, by bounding $\Heb zy$ using a similar argument, we have $\Heb xy\leq \Heb xz + \Heb zy \leq 2cn^2/\eps + 7\eps^2$.  	
	\end{poc}

	Returning to the reduction: we fix a starting vertex $u\in V(H)\subseteq V(G)$, the unvisited set to be $Y = V(H)\setminus\{u\}$, and set $X=V(G)\setminus Y$. Suppose that $H$ contains at least one Hamilton path $P$ starting at $u$. Then, by Claim \ref{clm1}, a strategy from $u$ following the path $P$ covers $Y$ in expected time at most $T_{\mathsf{Ham}} =(n-1)\cdot \left( (2c/\eps)n^2 + 7/\eps^2\right)$. Similarly, by Claim \ref{clm1}, it is clear that if $H$ does not contain a Hamilton path from $u$ then the time to cover $Y$ from $u$ is at least $n\cdot \left( (2c/\eps)n^2 - n\sqrt{c}/\eps^{3/2}\right)$, which is strictly greater than $T_{\mathsf{Ham}}$ for a suitable $c:=c(\eps)$.
	
	Now, if we are given $G(H)$ then we can run $\mathtt{Cost}(G,u,(V(G)\backslash V(H) )\cup \{u\})$ for each $u\in V(H)$ and if an answer less than $T_{\mathsf{Ham}}$ is returned for some $u$ then we conclude $H$ contains a Hamilton path from $u$. Otherwise we conclude that $H$ does not contain a Hamilton path.

	To prove the reduction to $\mathtt{BestStep}$, we make the following claim. \begin{clm}\label{clm2}If $H$ has a Hamilton path from $u$, then, in the walk started at $u$, any optimal next step from a vertex in $Y$ is to move towards the next unvisited vertex of $Y$ on some such path.
	\end{clm} \begin{poc} Suppose the walk is at $x\in Y$ the first time an optimal step chooses to move towards a vertex $y'\in Y$ which is not the next step in a Hamilton path from $u$. Since the expected remaining time decreases whenever an optimal step is taken, two successive optimal steps cannot be in opposite directions unless the walker visits an unvisited vertex in between. Thus the optimal strategy is to continue in the direction of $y'$ if possible. Such a strategy reaches $y'$ before returning to $x$ with probability at least $p:=2\eps/(1+\eps) -e^{C\cdot \eps cn^2}$, for some $C>0$ by \cite[Chapter XVI.2]{Feller}, and this takes at least $2cn^2$ steps. Suppose the walk reaches $y'$ before returning to $x$. Then, since $y'$ was not a vertex on any path extending the current list of visited vertices in $Y$ to a Hamilton path, the walk must revisit a previously visited vertex of $Y$. Thus, by Claim \ref{clm1}, such a strategy conditioned on the first step being in the direction of $y'$ has expected time at least $(n-1)\cdot \left( (2c/\eps)n^2 - n\sqrt{c}/\eps^{3/2}\right) + p\cdot 2cn^2>T_{\mathsf{Ham}}$ for a suitable  $c$.
	\end{poc}  
	
	Claim \ref{clm2} implies correctness of the following reduction for $\mathtt{BestStep}$: Given $G(H)$ and $u_0\in V(H)$, set $X_0= (V(G)\backslash V(H) )\cup \{u_0\}$ and call $\mathtt{BestStep}(G,u_0,X_0)$ to obtain a vertex $v$, which lies on the path from $u_0$ to $u_1$ for some $u_1\in V(H)$. Now set $X_1=X_0\cup\{u_1\}$ and call $\mathtt{BestStep}(G,u_1,X_1)$; proceed like this until $u_{n-1}$ has been identified. If $X_{n-1}=V(G)$, then $H$ has a Hamilton path from $u_0$ (and vice versa). We may therefore solve the Hamilton path problem by checking all start vertices $u_0\in V(H)$ in turn.\end{proof}

\end{document}